\documentclass[11pt]{article}

\usepackage{amsmath,amsfonts,amssymb,amscd,mathrsfs,amsthm} 
\usepackage{array}
\usepackage{bbm}
\usepackage{paralist}
\usepackage{xcolor}
\usepackage{multicol}
\usepackage{multirow}
\usepackage{bbm}
\usepackage{graphicx}
\usepackage{booktabs}
\usepackage{subfigure}
\usepackage{grffile}
\usepackage{float}   
\usepackage{mathrsfs}
\usepackage{cases}
\usepackage{lscape}
\usepackage{ulem}
\usepackage{soul,xcolor}

\newtheorem{proposition}{Proposition}
\newtheorem{assumption}{Assumption}
\newtheorem{lemma}{Lemma}

\newtheorem{remark}{Remark}
\newtheorem{theorem}{Theorem}[section]

\usepackage{mathtools}
\usepackage[font=small,labelfont=bf,labelsep=quad]{caption}
\usepackage[colorlinks,citecolor=blue,urlcolor=blue,linkcolor=blue]{hyperref}
\usepackage{setspace}
\usepackage[margin=1in]{geometry}

\usepackage{tikz}
\usetikzlibrary{arrows.meta, positioning, matrix}
\usepackage{amsmath, amssymb}

\renewenvironment{thebibliography}[1]{%
\begin{oldthebibliography}{#1}%
\setlength{\baselineskip}{.8em}
\linespread{1.5}
\small
\setlength{\parskip}{0ex}%
\setlength{\itemsep}{.05em}%
}%
{%
\end{oldthebibliography}%
}

\mathtoolsset{
  showonlyrefs,
  mathic 
}

\title{Optimal Contract Design with Quadratic Effort Cost }

\author{Xinfu Chen\footnote{The School of Mathematics, Southwestern University of Finance and Economics. peisw@swufe.edu.cn.}\qquad  Shuaijie Qian \footnote{Department of Mathematics, The Hong Kong University of Science and Technology. sjqian@ust.hk.}\qquad  Guan Qiao \footnote{Department of Mathematics, The Hong Kong University of Science and Technology. gqiaoaa@connect.ust.hk.}
}

\begin{document}

\maketitle

\begin{abstract}
The existence of an optimal contract of the principal-agent problem is a central issue in contract design. According to Cvitani{\'c} et al. \cite{cvitanic2009optimal}, such an optimal contract can be derived from the existence of a classical solution to the corresponding Hamilton–Jacobi–Bellman (HJB) equation, which is a degenerate, fully nonlinear parabolic equation. In this work, we follow their model, consider the case with drift control, and prove the existence of the classical solution to the HJB equation.
\end{abstract}

\section{Introduction}
In the effort model of the principal-agent problem, a principal delegates a project to an agent. The agent can exert effort to improve the output of the project. In order to motivate the agent to exert effort, the principal proposes a contract at the beginning. This contract is based on the project's output, since the principal can observe the output but not the effort levels. In view of this, the agent will exert effort subject to the trade-off between contract payoff and effort costs, while the principal will optimize the contract to balance the project output against the contract payoff. 

This Stackelberg game problem involves the optimization of the contract in a functional space. To circumvent this obstacle, 
Spear and Srivastava \cite{spear1987repeated} use recursive methods and treat the agent's continuation/promised utility as a state variable. This approach is later extended to continuous-time models by DeMarzo and Sannikov \cite{demarzo2006optimal} and  Sannikov \cite{sannikov2008continuous}.

Following this direction, Cvitani{\'c} et al. \cite{cvitanic2009optimal} characterize the optimal contract in a general continuous-time problem with quadratic cost, using the approach of backward stochastic differential equations. This contract is optimal under some assumptions about the existence of optima to a static optimization problem (see Proposition 3.1 of \cite{cvitanic2009optimal}), which is not further proven. 
After that, Cvitani{\'c} et al. \cite{cvitanic2018dynamic} further develop a systematic method to investigate the general principal-agent problem, which covers Cvitani{\'c} et al. \cite{cvitanic2009optimal}, by reformulating the principal's optimization problem as a standard control problem over both the output process and the promised utility process. The existence of an optimal contract is guaranteed if the value function of the standard control problem is sufficiently smooth (see Theorem 3.9 and Proposition 5.4 in \cite{cvitanic2018dynamic}). 

However, this regularity issue of the value function is {not} guaranteed in general by mathematicians. To the best of our knowledge, this type of regularity is only addressed in very special cases (see, e.g., explicit solutions in Cvitani{\'c} et al. \cite{cvitanic2018dynamic} or the one-dimensional case in Possama{\"\i} and Touzi \cite{possamai2025there}). 
This value function can be characterized by the corresponding HJB equation. However,
due to the special structure of the principal-agent problem, this partial differential equation (PDE) is degenerate and fully nonlinear, which is beyond the scope of classical theory for the existence of a classical solution, and usually only Sobolev regularity of $W^{2, 1}_p$ is guaranteed (see Krylov \cite{krylov1987nonlinear}, Chapter 7).

In this paper, by exploiting the special structure of the principal-agent problem with a quadratic cost function, we will construct a classical solution to this PDE problem. 
Given the uniqueness result from the comparison principle, we establish the regularity of the value function, and thus fill the gap in Cvitani{\'c} et al. \cite{cvitanic2009optimal} to verify that the conjectured optimal contract is indeed optimal. 

{\bf Related literature.} 
{Recently, the continuous-time principal-agent problem has received several remarkable extensions. Kr{\v{s}}ek and Possama{\"\i} \cite{krvsek2023randomisation} allow the agent to employ measure-valued controls so that the agent's continuation value can be characterized by a backward stochastic differential equation (BSDE) driven by a martingale measure. Under this setting, the principal's problem is reformulated as a weak control problem and compactification techniques are used to establish the existence of a measure-valued optimal contract. Hubert \cite{hubert2023continuous} adopts a second-order BSDE (2BSDE) to characterize the optimal contract for hierarchical scenarios that allow volatility control. Chiusolo and Hubert \cite{chiusolo2024new} demonstrate that the 2BSDE in \cite{cvitanic2018dynamic} can be simplified to a BSDE by assuming the principal can directly control the quadratic variation. Hern{\'a}ndez and Possama{\"\i} \cite{hernandez2024time} extend the framework of \cite{cvitanic2018dynamic} to time-inconsistent settings, showing that the principal's second-best problem is no longer a classical control problem but instead involves the control of a forward Volterra equation. Mastrolia and Zhang \cite{mastrolia2025agency} consider the mean-field system of agents managing projects with accidents. They characterize the mean-field equilibrium by the McKean-Vlasov equations. Moreover, they provide sufficient and necessary conditions for the optimal contract under weaker conditions than those in \cite{cvitanic2018dynamic}.}

{It is worth mentioning the development of fully nonlinear PDE theory. Wang \cite{wang1989regularity} establishes a priori estimates for viscosity solutions of fully nonlinear parabolic PDEs.
Additionally, PDE-based approaches to finance problems have drawn significant attention in the literature.
Guan et al. \cite{guan2023continuous} prove the existence of a classical solution to the fully nonlinear PDE arising from Markowitz's portfolio-selection problem under different borrowing and saving rates.
}

The structure of the paper is as follows. In Section \ref{sec model}, we introduce the  principal-agent problem. Section \ref{sec auxiliary} presents the (equivalent) standard control problem, proves the equivalence and derives the optimal contract. These results are all based on the assumption of the regularity of the value function. In Section \ref{sec regu w}, we establish the regularity of the value function, which constitutes the main contribution of this paper. Section \ref{conclusion} concludes the paper. 

\section{Model setup} \label{sec model}
In this section, we introduce the model setup of the principal-agent problem under consideration. In particular, we follow the model in Cvitani{\'c} et al. \cite{cvitanic2009optimal}. In Subsection \ref{sec agent}, we introduce the agent's problem for a given contract. In Subsection \ref{sec prin problem}, we introduce the principal's optimal contract design, taking into account the agent's response.    

Following Cvitani{\'c} et al. \cite{cvitanic2009optimal}, we assume that the market consists of a principal and an agent. 
The principal has a project which is delegated to the agent over the time interval $[0, T]$. The agent can control the output process by exerting effort. Precisely, we assume that the project's output process $\left\{ x(s) \right\}_{0\leq s \leq T}$ obeys the following stochastic differential equation (SDE):
\begin{equation}
\begin{cases}
    dx(s) = \sigma {\lambda}(s) ds + \sigma d \mathcal{B}(s), & \text{for } s \in (0, T], \\
    x(0) = x,
\end{cases}
\label{dyn_x}
\end{equation}
where \({\Lambda} = \{ {\lambda}(s) \}_{0 \leq s \leq T}\) is  the effort process chosen by the agent and  $\{\mathcal{B}(s)\}_{s \geq 0}$ is a standard one-dimensional Brownian motion on a filtered probability space 
$\big( \mathbb{S}, \mathscr{F}, \allowbreak \left\{\mathscr{F}_s\right\}_{s \geq 0}, P\big)$
with $\mathcal{B}(0) = 0$ almost surely. The filtration $\left\{\mathscr{F}_s\right\}_{s \geq 0}$ is
generated by the Brownian motion, is right-continuous and every $\mathscr{F}_s$ contains all $P$-null sets of $\mathscr{F}$.
\footnote{Following convention (see, e.g.,  Cvitani{\'c} et al. \cite{cvitanic2018dynamic}), the drift term is defined as $\sigma \lambda$ instead of $\lambda$.  As we will see in what follows, these two formulations are equivalent, up to a rescaling of the cost parameter $c$.}

The principal can observe the output process but not the agent's effort. In order to incentivize the agent to exert effort, the principal specifies a contract at time $0$ based on the terminal output to compensate for the agent's effort costs. Precisely, at time $T$, the principal will provide the agent with a payment of ${C}\left(x(T)\right)$, where $C$ is a function that satisfies the following assumption throughout this paper.    
\begin{assumption}\label{ass C}
    The reward function ${C}\in C^2(\mathbb{R})$ satisfies 
\begin{align}\label{growth C}
{\frac{1}{c_1} \leq C'(x)\leq c_1}, \quad {{-c_1 \leq C''(x) \leq c_1}}, \, \quad \forall x \in \mathbb{R}  
\end{align}    
for a constant $c_1 >1$.
\end{assumption}
We denote this set of functions by $\mathscr{C}$. Essentially, we require the contract to be monotonically increasing with linear growth.
This linear growth condition is natural 
because it is unlikely that the principal will pay more than the project's terminal value at time $T$. 

\subsection{Agent's problem} \label{sec agent}
Following Cvitani{\'c} et al. \cite{cvitanic2009optimal},
we assume that by exerting effort $\lambda(s)$, the 
agent's cost is $\frac{c}{2} \lambda^2(s)$ for a constant $c>0$.  
The agent's objective function is 
$$
J_A^{C;{\Lambda}}(x, t) := \mathbb{E}\left[U_A\left(C\left(x(T)\right)\right)-\frac{c}{2} \int_t^T {\lambda}^2(s) d s\bigg| x(t) = x\right],
$$
where $U_A$ is the agent's utility function.  
Moreover, we require $\Lambda$ to be chosen such that the SDE \eqref{dyn_x} permits a strong solution on $[t, T]$ and $\mathbb{E}\left[|U_A\left(C\left(x(T)\right)\right)|\right] <+\infty$. Denote this admissible set of controls by $\mathcal{A}_\lambda(x, t)$. 

The agent's problem is to find an optimal effort process $\left\{{\lambda}(s)\right\}_{t\leq s \leq T}$ to maximize his objective function, i.e.,
$$
v^C(x, t) := \sup _{\Lambda\in \mathcal{A}_\lambda(x, t)} J_A^{C;{\Lambda}}(x, t).
$$
\begin{assumption} \label{assumption Ua}
Agent's utility function $U_A(x) \in C^{{7}}(\mathbb{R})$ satisfies 
\begin{align}
    &\frac{1}{c_0} \leq U_A'(x) \leq c_0, \quad -c_0 \leq U_A''(x) \leq 0, \quad {-c_0 \leq U_A^{(i)}(x) \leq c_0, \quad 3\leq i\leq 7,}
\end{align}
for a constant $c_0 > 1$.
\end{assumption}

This assumption essentially requires $U_A$ to be asymptotically linear, which is important for deriving a bounded effort process $\lambda$. Based on this, we can establish the comparison principle invoked in Theorem \ref{theorem regu w}. 

{\bf Solution of agent's problem}

\begin{proposition}\label{prop v C2}
If the contract $C(x) \in \mathscr{C}$, 
then $v^C(x, t) \in C^{{\infty}}(\mathbb{R}\times [0, T)) \cap C(\mathbb{R}\times [0, T])$, and it is the unique classical solution to the following PDE problem 
\begin{equation}
\begin{cases}
	& v_t^C(x, t) + \frac{1}{2} \sigma^2 v_{xx}^C(x, t)  + \max \limits_{{\lambda}} \bigg\{ \sigma {\lambda} v_{x}^C(x, t) - \frac{c}{2} {\lambda}^2 \bigg\}   
	\\
    & \qquad = v_t^C(x, t) + \frac{\sigma^2}{2} \left( v_{xx}^C(x, t) + \frac{1}{c} \left[ v_x^C(x, t) \right]^2 \right) = 0, \quad \forall x \in \mathbb{R},\ 0\leq t < T, \\
    & v^C(x, T) = U_A(C(x)), \quad \forall x \in \mathbb{R}, 
\end{cases}
\label{pde_vc}
\end{equation}
satisfying the growth condition 
\begin{align}\label{equ growth v}
|v^C(x, t)| \leq D ( |x|+1), \ \forall x \in \mathbb{R},\ 0\leq t \leq T,\text{ for some constant $D>0$.}
\end{align}
The optimal effort can be derived by 
\begin{equation}
{\lambda}^C(s) = \frac{\sigma}{c} v_x^C(x(s), s),
\label{alpha_c}
\end{equation}
{and $\Lambda^C = \{\lambda^C(s)\}_{t \leq s \leq T} \in \mathcal{A}_\lambda (x,t)$.}
Moreover, there exists a constant $c_\lambda > 1$ such that 
\begin{align}\label{equ bd lambda}
\lambda^C(s) \in [\frac{1}{c_\lambda},c_\lambda].
\end{align}
\end{proposition}
{According to Proposition \ref{prop v C2},}
the corresponding output process $\left\{ x^C(s) \right\}_{t\leq s \leq T}$ is the solution to the SDE
\begin{equation}
\begin{cases}
    dx^C(s) = \sigma{\lambda}^C(s) ds + \sigma d \mathcal{B}(s), \\
    x^C(t) = x.
\end{cases}
\label{dyn_x_c}
\end{equation}
By setting
\begin{equation}
y^C(s)=v^C\left(x^C(s), s\right),
\label{y_c}
\end{equation}
we can calculate using It\^{o}'s lemma and \eqref{pde_vc} that 
\begin{equation}
\begin{aligned}
d y^C(s) & =\left[v_t^C\big(x^C(s), s\big)+\frac{\sigma^2}{2} v_{xx}^C\big(x^C(s), s\big)\right] d s+v_x^C\big(x^C(s), s\big) d x^C(s)\\
&= \left[v_t^C\big(x^C(s), s\big)+\frac{\sigma^2}{2} v_{xx}^C\big(x^C(s), s\big) + \sigma {\lambda}^C(s) v_x^C\big(x^C(s), s\big) \right] ds\\
&\quad + \sigma v_x^C\big(x^C(s), s\big) d \mathcal{B}(s)\\
& =\frac{c}{2} {{\lambda}^C(s)}^2 d s+c {\lambda}^C(s) d \mathcal{B}(s), \quad t \leqslant s \leqslant T.
\end{aligned}
\label{dyn_y_c}
\end{equation}

\subsection{Principal's problem}\label{sec prin problem}
The agent has an outside opportunity, which is modeled as reservation utility $y$. Thus, to ensure the agent's participation,  
the principal has to promise that the value of the agent's goal is no smaller than $y$. That is, 
\begin{align}\label{contract0}
y(t) = v^C(x, t) \geq y.  
\end{align}
Actually, given the monotonicity of utility functions, we only need to consider $ y(t) = y$. 

As elaborated in Section \ref{sec agent}, after the principal selects a reward function $C$, the agent chooses the optimal effort ${\Lambda}^C$, which generates the corresponding output process $\{x^C(s)\}_{t\leq s \leq T}$. The principal's objective function is
\begin{align}
\mathbb{E}\bigg[U_P\Big(x^C(T)-C\left(x^C(T)\right)\Big)\bigg]. \label{equ prin target}
\end{align}
where $U_P$ is the principal's utility function which satisfies the following Assumption \ref{assumption Up}. The argument of the utility function is the net value of the project after contract payoff.  
\begin{assumption}\label{assumption Up}
Principal's utility function $U_P(x)\in C^{{7}} (\mathbb{R})$ satisfies $U'_P(x) > 0 \geq \max \{U_P(x), U''_P(x) \}$, $\forall x \in \mathbb{R}. $ Additionally, we require 
\begin{align}
&{\frac{1}{c_0}  e^{-{c_0}{|x|}}  } \leq |U_P(x)| \leq c_0 ( 1 + e^{c_0{|x|}} ), 
\end{align}
for some constant $c_0 >0$
and
\begin{align}\label{equ UP derivartive}
-c_P \leq \frac{U_P'}{U_P} \leq -\frac{1}{c_P}, \quad 
\left|\frac{U_P^{(i)}}{U_P}\right| \leq c_P, \ 2\leq i \leq 7,
\end{align}
for $c_P >1$. 
\end{assumption}

Noticing that $y^C(T)=v^C\left(x^C(T), T\right)=U_A\left(C\left(x^C(T)\right)\right)$, it follows that
\begin{align}\label{equ uC}
&{C}\left(x^C(T)\right)=U_A^{-1}\left(y^C(T)\right),  
\end{align}
and thus noticing \eqref{contract0} and \eqref{equ prin target},  we define
$$u^C\left(x, y, t\right):=\mathbb{E}\bigg[U_P\big(x^C(T)-U_A^{-1}\left(y^C(T)\right)\big)\bigg| x^C(t) = x, y^C(t) = y\bigg]
$$
as the principal's value function under contract $C$. 

Define the set of contracts $\mathcal{C}(x, y, t)$ as those contracts which satisfy Assumption \ref{ass C} and \eqref{contract0} such that the above SDEs \eqref{dyn_x_c} and \eqref{dyn_y_c} both admit a strong solution in $[t, T]$ for any $x, y\in \mathbb{R}$. 
The principal's value function is defined as
$$u\left(x, y, t\right) : = \max \limits_{C\in \mathcal{C}(x, y, t)} u^C\left(x, y, t\right).$$ 

\section{Reduction to a standard stochastic control problem}\label{sec auxiliary}
The abovementioned principal-agent problem involves the optimization over a function $C(x)$, which is generally difficult. In the seminal work Cvitani{\'c} et al. \cite{cvitanic2018dynamic} and Cvitani{\'c} et al. \cite{cvitanic2009optimal}, this Stackelberg game problem is reformulated as a standard stochastic control problem.
In this section, for the paper's completeness, we introduce this stochastic control problem, and establish its equivalence to the original problem, subsequently deriving the candidate optimal contract.

Inspired by SDEs \eqref{dyn_x_c} and \eqref{dyn_y_c}, we let $\Lambda = \left\{{\lambda}(s)\right\}_{t\leq s \leq T}$ be a process such that {$\lambda(s) \in [\frac{1}{c_\lambda},c_\lambda]$ for some $c_\lambda > 1$}\footnote{We impose this bound on $\lambda$ by noticing \eqref{equ bd lambda}.} and the following system of SDEs admits a strong solution $\left\{ \left(x^{\lambda}(s), y^{\lambda}(s)\right)\right\}_{t \leq s \leq T}$ on $[t, T]$:
\begin{equation}
\begin{cases}
d x^{\lambda}(s)=\sigma {\lambda}(s) d s+\sigma d \mathcal{B}(s), \\
d y^{\lambda}(s)=\frac{c}{2} {\lambda}^2(s) d s+c {\lambda}(s) d \mathcal{B}(s), \\
x^{\lambda}(t)=x, \quad y^{\lambda}(t)=y. 
\end{cases}
\label{auxilliary}
\end{equation}
We denote the set of admissible controls by $\mathcal{A}(x, y, t)$ and define the value function
$$
w(x, y, t)=\sup _{\Lambda\in \mathcal{A}(x, y, t)} \mathbb{E}\left[U_P\big(x^{\lambda}(T)-U_A^{-1}(y^{\lambda}(T))\big)\right].
$$

Formally,
$w$ is a viscosity solution to the following PDE problem subject to an exponential growth condition 
\begin{align}\label{equ w visco}
\begin{cases}
w_t + \sup \limits_{{\lambda \in [\frac{1}{c_\lambda},c_\lambda]}} \bigg\{ \frac{\sigma^2}{2} w_{xx} + \sigma c {\lambda} w_{xy} + \frac{c^2 {\lambda}^2}{2} w_{yy} +  \sigma {\lambda} w_x + \frac{c}{2} {\lambda}^2 w_y \bigg\} = 0, \\
\qquad\qquad\qquad\qquad\qquad\qquad\qquad\qquad\qquad \forall x\in \mathbb{R},\ y\in \mathbb{R}, \ t \in [0, T), \ \\
w(x, y, T) = U_P \big( x - U_A^{-1}(y) \big),  \quad \forall x\in \mathbb{R}, y\in \mathbb{R}, \\
\max\limits_{t \in [0, T], x, y \in \mathbb{R}} \frac{w(x, y, t)}{e^{D (|x|+|y|) }} < \infty, \ \text{for some constant } D>0.
\end{cases}
\end{align} 
Assuming $w\in C^{2, 1}(\mathbb{R}^2 \times [0, T] )$, we define the function  
\begin{align}\label{equ defi L}
L(x, y, t):=w(x, y, t)+cw_y(x, y, t).
\end{align}
If 
\begin{align}\label{equ con L partial}
L_y<0,\  L_x>0 \text{ and } {\frac{L_x}{L_y} \in [-\frac{c}{\sigma}c_\lambda, -\frac{c}{\sigma}\frac{1}{c_\lambda}]}, 
\end{align}
then we can optimize the PDE operator in \eqref{equ w visco} to be 
\begin{equation}
   0 =  w_t + \frac{\sigma^2}{2} w_{xx} - \frac{\sigma^2}{2c} 
    \frac{\left(w_x + cw_{yx} \right)^2}{w_y + cw_{yy}} = w_t + \frac{\sigma^2}{2} w_{xx} -  \frac{\sigma^2}{2c}  \frac{L^2_x}{L_y} , \quad \forall (x, y) \in \mathbb{R}^2, \ t \in [0, T], \
\label{operator_w}
\end{equation}
and the optimal effort level is given by
\begin{equation} \label{def alpha_star}
{\lambda}^*(x, y, t) = \frac{-\sigma c w_{xy} - \sigma w_x}{c^2w_{yy} + cw_y} = -\frac{\sigma}{c} \frac{L_x(x, y, t)}{L_y(x, y, t)} \in [\frac{1}{c_\lambda},c_\lambda].
\end{equation}
The corresponding processes in \eqref{auxilliary} under the above optimal control are denoted by $\left\{ x^*(s), y^*(s) \right\}_{t\leq s \leq T}$, and the optimal strategy
$$
\lambda^*(s) = {\lambda}^*(x^*(s), y^*(s), s),\ t\leq s \leq T.
$$

Cvitani{\'c} et al. \cite{cvitanic2018dynamic} and Cvitani{\'c} et al. \cite{cvitanic2009optimal} have the following observation. 
\begin{proposition}\label{prop L}
Suppose $w\in C^{3, 2}(\mathbb{R}^2 \times [0, T])$\footnote{That is, $\partial^i_x \partial^j_y w$ and $\partial^k_x \partial^m_y \partial^n_t w$ are continuous for any $i, j, k, m, n \geq 0$ and $i+j \leq 3$, $k+m+n \leq 2$.},  $\{\lambda^*(s)\}_{t\leq s \leq T} \in \mathcal{A}(x, y, t)$, and \eqref{equ con L partial} holds. Then, along the path $\left\{ \left(x^*(s), y^*(s)\right) \right\}_{t\leq s \leq T}$, the function $L$
 is a constant. That is, 
\[
L\left(x^*(s), y^*(s), s\right) = L(x, y, t), \quad \forall s \in [t, T].
\]
Moreover, at time $T$ we have 
\begin{align}
L(x, y, T)&=w(x, y, T)+cw_y(x, y, T)\\
&= U_P\left(x-U_A^{-1}(y)\right)-c U_P^{\prime}\left(x-U_A^{-1}(y)\right) (U_A^{-1})^{\prime}(y).
\label{eq L^T}
\end{align}
\end{proposition}
\begin{remark}\label{rmk 2d sur}
It follows from this proposition that $\left(x^*(s), y^*(s), s\right)$ evolves on a two-dimensional surface, which is characterized by $L(x, y, s) = const$.  
Once the equivalence between the control problem and the principal-agent problem is established, we see that $y^*(s) = v^{C^*}(x^*(s), s)$ under the optimal contract $C^*$, which exactly corresponds to this surface. 
\begin{figure}[H]
        \centering
        \includegraphics[width=0.5\textwidth]{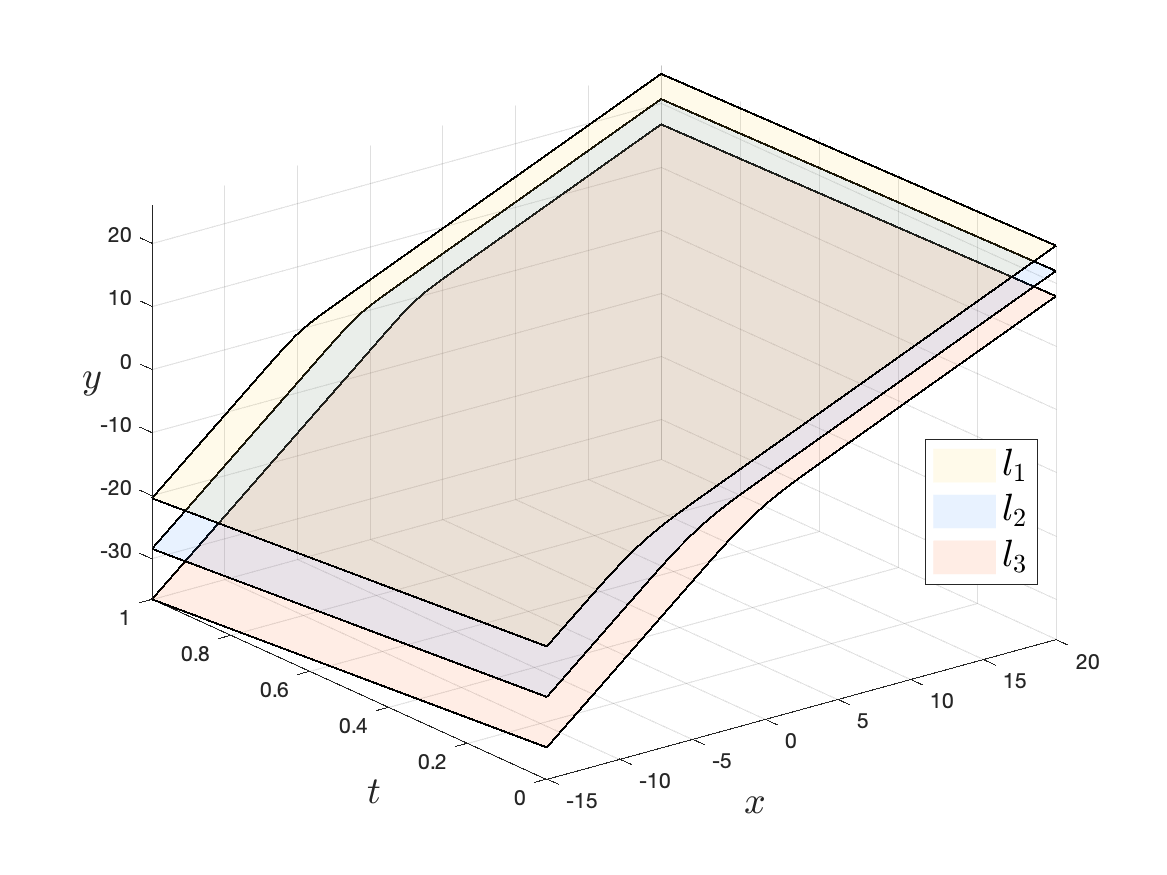}
        \caption{Two-dimensional surfaces $\mathcal{S}(l)$.}
       \label{2_d_surface}
\end{figure}
Denote $\mathcal{S}(l): = \bigg\{(x, y, s) \in \mathbb{R}^2 \times [0, T]\bigg| L(x, y, s) = l  \bigg\}$, then we see $\cup_{l \in \mathbb{R}} \mathcal{S}(l) = \mathbb{R}^2 \times [0, T]$ and $\mathcal{S}(l_1)\cap \mathcal{S}(l_2) = \emptyset$, $\forall l_1\neq l_2$. Moreover, due to  Proposition \ref{prop L}, $\mathcal{S}(l)$ is non-empty if and only if $l$ is in the image of function $L(x, y, T)$, which is an open set and we denote it as $\mathcal{L} = (-\infty, 0)$. 
\end{remark}

\subsection{Equivalence between the two problems and the optimal contract} \label{sec equi uw}
Suppose the principal-agent problem and the stochastic control problem are equivalent, then given $x(t) = x, y(t) = y$, we can derive $l = L(x, y, t)$. It follows from Proposition \ref{prop L} that
\begin{align}
l = L\left(x^*(T), y^*(T), T\right).
\end{align}
Since the two problems are equivalent, intuitively we have $$x^*(T) = x^{C^*}(T),\ y^*(T)=y^{C^*}(T) $$
for the optimal contract $C^*$, and thus $$y^*(T) = y^{C^*}(T) = U_A\Big(C^*(x^{C^*}(T))\Big) =U_A\Big(C^*(x^*(T))\Big). $$
Thus, for the constant $l= L(x, y, t)$, the optimal contract $C^*(x; l)$ is implicitly derived by 
\begin{align}
l = & L\bigg(x, U_A\Big(C^*\big(x; l\big)\Big), T\bigg)
= U_P\big(x-C^*(x; l)\big)-c \frac{U_P^{\prime}\big(x-C^*(x; l)\big)}{U_A^{\prime} ( C^*(x; l) )}. \label{equ C*0}
\end{align}

In the following theorem, we derive two results. 
On the one hand, we establish the equivalence between the principal-agent problem in Section \ref{sec model} and the stochastic control problem in Section \ref{sec auxiliary}. Cvitani{\'c} et al. \cite{cvitanic2018dynamic} also prove a similar result in the language of stochastic processes. In contrast, our proof is based on PDE. On the other hand, we also prove that the contract derived in \eqref{equ C*0} is indeed optimal. 
\begin{theorem}\label{thm equi}
Suppose $w(x, y, t) \in C^{3, 2} (\mathbb{R}^2 \times [0, T])$ and \eqref{equ con L partial},
then we have the following results. \\
(i). $u(x, y, t) = w(x, y, t)$, $\forall (x, y, t)\in \mathbb{R}^2 \times [0, T]$.\\
(ii). For any $x \in \mathbb{R}, l \in \mathcal{L}$, there is a unique solution $C^*(x; l)$ to  \eqref{equ C*0} and 
it is the optimal contract. 
\end{theorem}

The following theorem verifies the regularity assumption in Theorem \ref{thm equi}, and it constitutes the main contribution of our paper.
Precisely, we step-by-step construct a $C^{5, 2}$ solution to \eqref{operator_w} subject to the corresponding boundary and growth conditions. After that, we establish a comparison principle for the PDE, ensuring that the constructed solution is indeed the value function. 
\begin{theorem}\label{theorem regu w}
The function $w(x, y, t) \in C^{5,2}(\mathbb{R}^2\times [0, T])$ and  
\begin{align}\label{equ con Lx Ly}
    L_y(x, y, t)<0<L_x(x, y, t),\  {\frac{L_x}{L_y} \in [-\frac{c}{\sigma}c_\lambda, -\frac{c}{\sigma}\frac{1}{c_\lambda}]},\ \forall x, y \in \mathbb{R},\ t \in [0, T].
\end{align}
Moreover, $w(x, y, t)$ is the unique classical solution to the PDE \eqref{equ w visco}.
\end{theorem}

\section{Proof of Theorem \ref{theorem regu w} }\label{sec regu w}
{\bf Step 1}: Motivated by Remark \ref{rmk 2d sur}, we recover the value function on each surface defined by $L(x, y, t) = l$. 
Precisely, we show that the solution $\hat{Y}(x, l, t)$ to the following PDE exists in the space {$C^{{6,3}} (\mathbb{R}\times \mathcal{L}\times[0, T])$ with at most linear growth in $x$ for each fixed $l$.}
\begin{align}
\begin{cases}
    \hat{Y}_t + \frac{\sigma^2}{2} \hat{Y}_{xx} + \frac{\sigma^2}{2c} \hat{Y}_x^2 = 0, & \forall (x, l, t) \in \mathbb{R}\times \mathcal{L} \times [0, T), \\
    \hat{Y}(x, l, T) = \hat{Y}^T(x, l), & \forall x \in \mathbb{R}, l \in \mathcal{L}.
\end{cases}
\label{equ hat Y}
\end{align}
Here,  $\hat{Y}^T(x, l)$ is defined implicitly by
\begin{align}
l=L(x, \hat{Y}^T, T).
\label{equ term hatY}
\end{align}
{This PDE problem is closely related to \eqref{pde_vc} by noticing that $\hat{Y}^T(x, l)$ is nothing but $U_A\Big(C^*(x; l)\Big)$. That is, we select the candidate optimal contract, and $\hat{Y}$ is the agent's value function. 
}

\vspace{8pt}
{\bf Step 1.1}. We prove that the terminal condition $\hat{Y}^T(x, l)$ is well-defined and is in $C^6(\mathbb{R}\times \mathcal{L})$. 

\begin{lemma}\label{lem mono L}
We have $L_y(x, y, T) < 0 < L_x(x, y, T)$, $\forall x\in \mathbb{R},\  y\in \mathbb{R}$.
\end{lemma}
\begin{proof}[Proof of Lemma \ref{lem mono L}]
From \eqref{eq L^T}, we can directly calculate
\begin{align*}
L_y(x, y, T) &= - U_P' \big(x - U_A^{-1}(y)\big) \frac{d}{dy} U_A^{-1}(y) + c U_P''(x - U_A^{-1}(y)) \left(\frac{d}{dy} U_A^{-1}(y)\right)^2 \\
&\quad - c U_P'\big(x - U_A^{-1}(y)\big) \frac{d^2}{dy^2} U_A^{-1}(y).  
\end{align*}
We denote $ z = U_A^{-1}(y)$, then direct calculation shows 
\begin{align*}
\frac{d}{dy} U_A^{-1}(y)  = \frac{1}{U_A'(z)}>0, \quad \frac{d^2}{dy^2} U_A^{-1}(y)= -\frac{U_A''(z)}{U_A'(z)} \Big(\frac{d}{dy} U_A^{-1}(y)\Big)^2 = -\frac{U_A''(z)}{U_A'(z)^3}\geq0.
\end{align*}
Assumption \ref{assumption Ua} implies $U_A' > 0,\ U_A'' \leq 0$. We hence obtain $L_y(x, y, T) < 0$. Similarly, we have
\begin{align*}
L_x(x, y, T) = U'_P\left(x-U_A^{-1}(y)\right)-c U_P^{\prime\prime}\left(x-U_A^{-1}(y)\right) \frac{d}{dy}U_A^{-1}(y) >0.
\end{align*}
\end{proof}

The following lemma proves the existence and characterizes the properties of the function $\hat Y^T$. 
\begin{lemma}\label{lem term Y}
For any $(x, l)\in \mathbb{R}\times \mathcal{L}$, there is a unique $\hat{Y}^T(x, l)$ such that  \eqref{equ term hatY} holds.  Moreover, we have the following properties about $\hat Y^T$.\\
(i). $\hat{Y}^T(x, l) \in C^{{6}}(\mathbb{R}\times \mathcal{L})$ and there exists a $D(l)>0$, such that
$|\hat{Y}^T(x, l)|\leq D(l)|x| + D(l)$ and $D(l)$ is locally uniformly bounded for $l \in \mathcal{L}$. \\
(ii). $U_P\left(x-U_A^{-1}(\hat{Y}^T(x, l))\right)\sim l<0$. \footnote{That is, there exists a constant $C>1$, independent of $(x, l)$, such that $\frac{l}{U_P\left(x-U_A^{-1}(\hat{Y}^T(x, l))\right)}\in [1/C, C]$.}\\
(iii). {$\hat{Y}_x^T(x, l) \in [1/C, C]$ and $\hat{Y}_{xx}^T(x, l) \in [-C, C]$ for some global constant $C>0$ independent of $(x, l)$.  
Moreover, $\hat{Y}_l^T(x,l) \sim \frac{1}{l}$.}\\
(iv). There exist a finite constant $d > 0$ and a polynomial function $g(\cdot)$ of finite degree, such that
for all integers $i,j \geq 0$ and $1\leq i+j \leq 6$, $\left|\frac{\partial^{i+j} \hat{Y}^T(x, l)}{\partial x^i \partial l^j}\right| \leq \frac{g(|l|)}{|l|^d}$. \\
\end{lemma}
\begin{proof}[Proof of Lemma \ref{lem term Y}]
We will only prove the existence and property $(i)$ here. Proofs of the other properties are relegated to the appendix. 

Since $l \in \mathcal{L}$, there exists $(x_0, y_0)$ such that $L(x_0, y_0, T) = l$. We then consider general $x$ and  solve for 
 \begin{align}
  l =& L(x, \hat{Y}^T(x,l), T) \\
  =& U_P\left(x-U_A^{-1}(\hat{Y}^T)\right)-c U_P^{\prime}\left(x-U_A^{-1}(\hat{Y}^T)\right) (U_A^{-1})^{\prime}(\hat{Y}^T)\\
  = &  U_P\left(x-U_A^{-1}(\hat{Y}^T)\right)-c U_P^{\prime}\left(x-U_A^{-1}(\hat{Y}^T)\right) \frac{1}{U_A'(U_A^{-1}(\hat{Y}^T))}.
  \end{align}
  
  Given $L_y<0<L_x$ from Lemma \ref{lem mono L}, 
 it follows that for any $x> x_0$, $L(x, y_0, T) > L(x_0, y_0, T) = l$. Furthermore, given that $U'_A>0\geq U''_A$, we have 
 \begin{align}
      & L\bigg(x, U_A\big( U_A^{-1}(y_0)+(x-x_0)\big), T\bigg)\\
  = &U_P\left(x_0-U_A^{-1}(y_0)\right)-c U_P^{\prime}\left(x_0-U_A^{-1}(y_0)\right) \frac{1}{U_A'\bigg(U_A^{-1}(y_0) +(x-x_0)\bigg)}\\
  \leq & U_P\left(x_0-U_A^{-1}(y_0)\right)-c U_P^{\prime}\left(x_0-U_A^{-1}(y_0)\right) \frac{1}{U_A'\bigg(U_A^{-1}(y_0) \bigg)}\\
  = &l. 
 \end{align}
 Therefore, for any $x>x_0$, there exists a unique solution
 \begin{align}\label{equ v* bdu}
 \hat{Y}^T(x,l)\in \bigg(y_0, U_A\big( U_A^{-1}(y_0)+(x-x_0)\big)\bigg].
 \end{align}
Analogously, for any $x<x_0$,  there exists a unique solution 
  \begin{align}\label{equ v* bdl}
  \hat{Y}^T(x,l)\in \bigg[U_A\big( U_A^{-1}(y_0)-(x_0-x)\big), y_0\bigg).
  \end{align}
Consequently, existence, uniqueness and the growth condition are guaranteed. {Moreover, the growth condition is {locally} uniform for $l\in\mathcal{L}$.}
By Implicit Function Theorem and Lemma \ref{lem mono L}, it follows that $Y^T(x, l)\in {C^{6}(\mathbb{R}\times \mathcal{L})}$ provided that $U_A, U_P\in {C^{7}(\mathbb{R})}$.
\end{proof}

\vspace{8pt}

{\bf Step 1.2}. We show that the solution $\hat{Y}(x, l, t)$ exists and belongs to {$C^{6,3} (\mathbb{R} \times \mathcal{L} \times[0, T]).$}

Set $\hat{m}(x, l, t)=e^{\frac{1}{c} \hat{Y}(x, l, t)}$,
then \eqref{equ hat Y} is equivalent to 
\begin{equation} \label{equ mhat}
    \begin{cases}
        \hat{m}_t+\frac{\sigma^2}{2} \hat{m}_{x x}=0 , &\text{for } (x, l, t) \in \mathbb{R}\times \mathcal{L} \times [0, T),\\
        \hat{m}(x, l, T)=e^{\frac{1}{c} \hat{Y}^T(x, l)}, & \text{for } (x, l) \in \mathbb{R}\times \mathcal{L}.
    \end{cases}
\end{equation}

To complete this step, we need to guarantee the growth condition of the terminal value function $\hat{m}(x, l, T)$ and the derivatives $$\left|\frac{\partial^{i+j} \hat{m}(x, l, T)}{\partial x^i \partial l^j}\right|, \text{ for $i+j \leq 6$}.$$
The above {Lemma \ref{lem term Y} (i)} implies the exponential growth of the terminal value function $\hat{m}(x, l, T)$ in $x$ for any fixed $l$. {Lemma \ref{lem term Y} (iii) and (iv)} further implies that $\left|\frac{\partial^{i+j} \hat{m}(x, l, T)}{\partial x^i \partial l^j}\right|$ is bounded by $C(l) \hat{m}(x, l, T)$, where $C(l)$ is a constant continuously depending on $l$. 

Therefore, for any fixed $l$, the solution $\hat m$ to \eqref{equ mhat} exists. For partial derivatives $f(x, l, t) : =  \frac{\partial^{i+j} \hat{m}(x, l, t)}{\partial x^i \partial l^j}$, we see from \eqref{equ mhat} that it satisfies 
\begin{equation} 
    \begin{cases}
        f_t+\frac{\sigma^2}{2} f_{x x}=0 , &\text{for } (x, l, t) \in \mathbb{R}\times \mathcal{L} \times [0, T),\\
        f(x, l, T)=\frac{\partial^{i+j} \hat{m}(x, l, t)}{\partial x^i \partial l^j}, & \text{for } (x, l) \in \mathbb{R}\times \mathcal{L}.
    \end{cases}
\end{equation}
From the growth condition and Lemma \ref{lemma_conts_y}, this function $f$
is continuous. 

\vspace{8pt}
{\bf Step 2:} We show that for any $(x, y, t) \in \mathbb{R}^2 \times [0, T]$, there exists a unique $\hat{L}(x, y, t)\in {C^{{6,3}}(\mathbb{R}^2 \times [0, T])}$ such that
\begin{align}\label{equ defi hatL}
y=\hat{Y}(x, \hat{L}(x, y, t), t).
\end{align}
Moreover, 
\begin{align}\label{equ LxLy sign}
    \hat{L}_y <0 < \hat{L}_x, \forall (x, y, t) \in\mathbb{R}^2 \times [0, T].
\end{align}
 
This $\hat{L}$ function entails the information of the surfaces in Figure \ref{2_d_surface}
and it will be used to construct a classical solution in Step 4.

{First, Lemma \ref{lem mono L} implies $\hat{Y}^T(x, l)$ is strictly decreasing in $l$. Therefore, 
\begin{align}\label{equ hatyl}
\hat{Y}_l(x, l, t)<0,\ \forall (x, t) \in \mathbb{R}\times [0, T]
\end{align} 
by the comparison principle for the parabolic PDE \eqref{equ mhat}.
}

Next, 
we take $\partial_x$ on both sides of
$$ l = L(x, \hat{Y}^T(x, l), T) $$
and derive that 
$$ 0 = L_x + L_y \hat{Y}^T_x. $$
From Lemma \ref{lem mono L} we derive $\hat{Y}^T_x(x, l)>0$. Then by the comparison principle for \eqref{equ mhat}, we obtain
\begin{align}\label{equ hatyx}
\hat{Y}_x(x, l, t)>0,\ \forall (x, t) \in \mathbb{R}\times [0, T].
\end{align}

Thirdly, we show the existence and regularity of $\hat{L}$. From \eqref{equ hatyl} and the definition, it follows that 
$\lim\limits_{l\to \inf \mathcal{L}} \hat{Y}^T(x, l) = +\infty$ for any fixed $x$. Therefore, given \eqref{equ mhat}, $\lim\limits_{l\to \inf \mathcal{L}}\hat{Y}(x, l, t) = +\infty$ for any fixed $(x, t)$. 
Similarly we have $\lim\limits_{l\to \sup \mathcal{L}}\hat{Y}(x, l, t) = -\infty$ for any fixed $(x, t)$.
Existence is thus ensured by the Intermediate Value Theorem, and the $C^{{6,3}}$ regularity follows from the Implicit Function Theorem.  

We finally show that $\hat{L}_y <0 < \hat{L}_x$. On the one hand, since $\hat{Y}_l<0$, it follows from \eqref{equ defi hatL} that $\hat{L}_y <0$. On the other hand, taking the partial derivative $\partial_x$ in \eqref{equ defi hatL}, we obtain
$$ 0 = \hat{Y}_x+\hat{Y}_l \hat{L}_x. $$
In view of \eqref{equ hatyl} and \eqref{equ hatyx}, we conclude that $\hat{L}_x>0. $

\vspace{8pt}
{\bf Step 3: } 
Since $\hat{Y}$ satisfies the PDE \eqref{equ hat Y}, $\hat{L}(x, y, t)$ will correspondingly satisfy a PDE. 
We derive this PDE problem for later use in Step 4 to verify the constructed function is a solution to the PDE \eqref{equ w visco}.

For any $(x, l, t) \in \mathbb{R}\times\mathcal{L} \times [0, T)$, we compute the following partial derivatives:
\begin{align*}
& \frac{\partial}{\partial x} \eqref{equ defi hatL} \Rightarrow \quad 0=\hat{Y}_x+\hat{Y}_l \hat{L}_x, \\
& \frac{\partial}{\partial t}  \eqref{equ defi hatL}\Rightarrow \quad 0=\hat{Y}_t+\hat{Y}_l \hat{L}_t, \\
& \frac{\partial}{\partial y}  \eqref{equ defi hatL} \Rightarrow \quad 1=\hat{Y}_l \hat{L}_y, \\
& \frac{\partial^2}{\partial x^2} \eqref{equ defi hatL}  \Rightarrow \quad 0=\hat{Y}_{x x}+2 \hat{Y}_{x l} \hat{L}_x+\hat{Y}_{l l} \hat{L}_x^2+\hat{Y}_l \hat{L}_{x x}, \\
& \frac{\partial^2}{\partial y^2}   \eqref{equ defi hatL}\Rightarrow \quad 0=\hat{Y}_l \hat{L}_{y y}+\hat{Y}_{l l} \hat{L}_y^2, \\
& \frac{\partial^2}{\partial x \partial y} \eqref{equ defi hatL} \Rightarrow \quad 0=\hat{Y}_{x l} \hat{L}_y+\hat{Y}_{l l} \hat{L}_x \hat{L}_y+\hat{Y}_l \hat{L}_{x y}.
\end{align*}

From above, we obtain:
\begin{align} 
	\hat{Y}_l & =\frac{1}{\hat{L}_y}, \quad \hat{Y}_t=-\frac{\hat{L}_t}{\hat{L}_y}, \quad \hat{Y}_x=-\frac{\hat{L}_x}{\hat{L}_y}, \label{deri_YtoL}\\ \hat{Y}_{l l} & =-\frac{\hat{L}_{y y}}{\hat{L}_y^3}, \\ 
	\hat{Y}_{x l} & =-\hat{L}_x \hat{Y}_{l l}-\frac{\hat{Y}_l}{\hat{L}_y} \hat{L}_{x y}=\frac{\hat{L}_x}{\hat{L}_y^3} \hat{L}_{y y}-\frac{\hat{L}_{x y}}{\hat{L}_y^2}, \\ 
	\hat{Y}_{x x} & =-2 \hat{Y}_{x l} \hat{L}_x-\hat{Y}_{l l} \hat{L}_x^2-\hat{Y}_l \hat{L}_{x x} \\ 
	& =2 \hat{L}_x\left(-\frac{\hat{L}_x}{\hat{L}_y^3} \hat{L}_{y y}+\frac{1}{\hat{L}_y^2} \hat{L}_{x y}\right)+\frac{\hat{L}_x^2}{\hat{L}_{y}^3} \hat{L}_{y y}-\frac{1}{\hat{L}_y} \hat{L}_{x x} \\ 
	& =\frac{1}{\hat{L}_y}\left(-\hat{L}_{x x}-\frac{\hat{L}_x^2}{\hat{L}_y^2} \hat{L}_{y y}+2 \frac{\hat{L}_x}{\hat{L}_y} \hat{L}_{x y}\right).\end{align}
Thus, \eqref{equ hat Y} becomes
\begin{align}
0 = \hat{Y}_t+\frac{\sigma^2}{2} \hat{Y}_{x x}+\frac{\sigma^2}{2 c} \hat{Y}_x^2 =\frac{1}{\hat{L}_y}\left[-\hat{L}_t-\frac{\sigma^2}{2}\left(\hat{L}_{x x}-2 \frac{\hat{L}_x}{\hat{L}_y} \hat{L}_{x y}+\frac{\hat{L}_x^2}{\hat{L}_y^2} \hat{L}_{y y}\right)+\frac{\sigma^2}{2c} \frac{\hat{L}_x^2}{\hat{L}_y}\right]. 
\end{align}
Since $\hat{L}_y<0$ from Step 2, 
we conclude that $\hat{L}$ satisfies
\begin{equation}\label{equ PDE hatL}
    \begin{cases}
        \hat{L}_t+\frac{\sigma^2}{2}\left(\hat{L}_{x x}-\frac{2 \hat{L}_x}{\hat{L}_ y} \hat{L}_{x y}+\frac{\hat{L}_x^2}{\hat{L}_ y^2} \hat{L}_{y y}\right)-\frac{\sigma^2}{2 c} \frac{\hat{L}_x^2}{\hat{L}_y}=0, \quad \forall (x, y, t) \in \mathbb{R}^2 \times[0, T), \\
        \hat{L}(x, y, T)=L(x, y, T),\quad \forall (x, y) \in \mathbb{R}^2,
    \end{cases}
\end{equation}
where $L(x, y, T)$ is given by \eqref{eq L^T}.

\vspace{8pt}
{\bf Step 4: }
Finally, given the function $\hat L$ derived in Step 2, we set $\hat{w}(x, y, t)$ as the solution to
\begin{equation}
\begin{cases} \label{PDE hatw}
\hat{w}_t+\frac{\sigma^2}{2} \hat{w}_{x x}=\frac{\sigma^2}{2 c} \frac{\hat{L}_x^2}{\hat{L}_y},\quad \text{for } (x, y, t) \in \mathbb{R}^2 \times[0, T), \\
\hat{w}(x, y, T)=U_P\left(x-U_A^{-1}(y)\right),\quad \text{for } (x, y) \in \mathbb{R}^2.
\end{cases}
\end{equation}

We show that this PDE problem admits a unique solution $\hat{w} \in C^{5, 2}(\mathbb{R}^2 \times [0, T])$ with at most exponential growth. Subsequently, we prove that $\hat{w}(x, y, t) = w(x, y, t)$, for $ (x, y, t) \in \mathbb{R}^2 \times[0, T]$ and that the function $L$ satisfies 
\begin{align}
L_y<0<L_x, \text{ and } \frac{L_x}{L_y} \in [-\frac{c}{\sigma}c_\lambda, -\frac{c}{\sigma}\frac{1}{c_\lambda}]. 
\end{align}
The main difficulty here is that while the PDE is solved in $(x, t)$, we require regularity in $(x, y, t)$. 

\vspace{8pt}
{\bf Step 4.1}. We prove the existence of a solution $\hat{w}\in C^{5, 2}(\mathbb{R}^2 \times [0, T])$ with at most exponential growth. The following lemma is essential for this target. 
\begin{lemma}\label{lemma grow L} 
For integers $i,j \geq 0$ and $1 \leq i+j \leq 5$,
    \begin{align*}
    & {\left| \hat{L} \right|},
    {\bigg|\frac{(\hat{L}_x)^2}{\hat{L}_y}\bigg|}, 
    {\left|\frac{\partial^{i+j} \left(\frac{(\hat{L}_x)^2}{\hat{L}_y}\right)}{\partial x^i \partial y^j}\right|}, 
    {\left|\frac{\partial \left(\frac{(\hat{L}_x)^2}{\hat{L}_y}\right)}{\partial t}\right|}, 
    {\left|\frac{\partial^2 \left(\frac{(\hat{L}_x)^2}{\hat{L}_y}\right)}{\partial t^2}\right|}, 
    {\left|\frac{\partial^2 \left(\frac{(\hat{L}_x)^2}{\hat{L}_y}\right)}{\partial t \partial x}\right|}, 
    {\left|\frac{\partial^2 \left(\frac{(\hat{L}_x)^2}{\hat{L}_y}\right)}{\partial t \partial y}\right|}\\
     &\leq c_2 ( 1 + e^{c_2 ({|x|+|y|})}), 
    \end{align*}
    for a constant $c_2 \in (0, \infty)$ which is independent of $(x, y, t)$. 
\end{lemma}
In what follows, we establish the existence of solution $\hat{w}(x, y, t) \in {C^{5,2}(\mathbb{R}^2 \times [0, T])}\subset {C^{3,2}(\mathbb{R}^2 \times [0, T])}$ from Lemma \ref{lemma grow L} and {Lemma \ref{lemma_conts_y}}.

On the one hand, Lemma \ref{lemma grow L} ensures that the function $\hat{w}$ is well-defined and satisfies
\begin{align} \label{growth_hat_w}
\max\limits_{t \in [0, T], x, y \in \mathbb{R}} \frac{\hat {w}(x, y, t)}{e^{D (|x|+|y|) }} < \infty, \ \text{for some constant } D>0.
\end{align}

On the other hand, the regularity is guaranteed by noticing that $\hat{L}(x, y, t)\in {C^{{6,3}} (\mathbb{R}^2 \times [0, T])}$ from Step 2. The function $f(x, y, t)$ in Lemma \ref{lemma_conts_y} can be set as partial derivatives of $\hat{w}$. 

Precisely, first, for any $0\leq i+j \leq 5$, we define 
\begin{align}
    f(x, y, t) &:= \frac{\partial^{i+j}}{\partial x^i \partial y^j}\hat{w}(x, y, t),\\
    H(x, y, t) &:=  \frac{\partial^{i+j}}{\partial x^i \partial y^j} \big(\frac{\hat{L}_x^2}{\hat{L}_y}\big)(x, y, t),\\
    U(x, y) &:= \frac{\partial^{i+j}}{\partial x^i \partial y^j} \big( U_P\left(x-U_A^{-1}(y)\right) \big).
\end{align}
Then $f$ satisfies \eqref{PDE_w&H}, it follows from Lemma \ref{lemma grow L} and Lemma \ref{lemma_conts_y} that $f$ is continuous.

Second, define $f = \hat{w}_t$. From \eqref{PDE hatw}, for $t< T$, we have 
\begin{align}
    f_t + \frac{\sigma^2}{2} f_{xx} = \partial_t \big(\frac{\sigma^2}{2c}\frac{\hat{L}_x^2}{\hat{L}_y} \big).
\end{align}
At terminal time $T$, \eqref{PDE hatw} yields
\begin{align}\label{equ wt term}
f(x, y, T) =\frac{\sigma^2}{2c} \frac{\hat{L}_x^2}{\hat{L}_y}(x, y, T) - \frac{\sigma^2}{2} \big( U_P\left(x-U_A^{-1}(y)\right) \big)_{xx}.  
\end{align}
By Lemma \ref{lemma grow L} and Lemma \ref{lemma_conts_y}, we conclude that $f$ is continuous. Similarly, these arguments work for $\hat{w}_{tx}$ and  $\hat{w}_{ty}$. 

Finally, define $f = \hat{w}_{tt}$, we have
\begin{align}
    f_t + \frac{\sigma^2}{2} f_{xx} = \partial_{tt} \big(\frac{\sigma^2}{2c}\frac{\hat{L}_x^2}{\hat{L}_y} \big),\ \forall t <T.
\end{align}
At $t = T$,  it follows from \eqref{PDE hatw} that
\begin{align}
   (\hat{w}_t)_t + \frac{\sigma^2}{2} (\hat{w}_t)_{xx} = \partial_{t} \big(\frac{\sigma^2}{2c}\frac{\hat{L}_x^2}{\hat{L}_y} \big). 
\end{align}
From \eqref{equ wt term}, we then derive 
\begin{align}
f(x, y, T) = &\partial_{t} \big(\frac{\sigma^2}{2c}\frac{\hat{L}_x^2}{\hat{L}_y} \big) - \frac{\sigma^2}{2} \big( \hat{w}_t \big)_{xx}\bigg|_{(x, y, T)}\\  
= & \partial_{t} \big(\frac{\sigma^2}{2c}\frac{\hat{L}_x^2}{\hat{L}_y} \big) - \frac{\sigma^2}{2} \big(\frac{\sigma^2}{2c} \frac{\hat{L}_x^2}{\hat{L}_y}\big)_{xx}\bigg|_{(x, y, T)} + \frac{\sigma^4}{4} \partial_x^4 \bigg( U_P\left(x-U_A^{-1}(y)\right) \bigg).  
\end{align}
Applying Lemma \ref{lemma grow L} and Lemma \ref{lemma_conts_y} once more, we conclude that $f$ is continuous. 

\vspace{8pt}
{\bf Step 4.2}. We show $\hat{w}(x, y, t) = w(x, y, t)$ and verify the condition \eqref{equ con Lx Ly}.  

\vspace{8pt}
(i). We first derive that $\hat w$ satisfies the following PDE 
\begin{equation}
\begin{cases}
    \hat{w}_t + \frac{\sigma^2}{2} \hat{w}_{xx} - \frac{\sigma^2}{2c} 
    \frac{\left(\hat{w}_x + c\hat{w}_{yx} \right)^2}{\hat{w}_y + c\hat{w}_{yy}} = 0, \quad \forall (x, y) \in \mathbb{R}^2, \ t \in [0, T], \\
    \hat{w}(x, y, T)=U_P\left(x-U_A^{-1}(y)\right),\quad \forall (x, y) \in \mathbb{R}^2.
\label{finalpdew}
\end{cases}
\end{equation}
Taking $\partial_y$ of \eqref{PDE hatw} on both sides, we obtain

\begin{align}\label{equ pde wy}
\left(\hat{w}_y\right)_t+\frac{\sigma^2}{2}\left(\hat{w}_y\right)_{x x}  =\frac{\sigma^2}{2 c} \frac{2 \hat{L}_{x} \hat{L}_{y} \hat{L}_{xy}-\hat{L}_x^2 \hat{L}_{y y}}{\hat{L}_y^2}.
\end{align}
Set $\tilde{L}(x, y, t)=\hat{w}(x, y, t)+c \hat{w}_y(x, y, t)$, {then $\tilde{L}(x, y, t)$ satisfies the growth condition \eqref{growth_hat_w} by noticing Lemma \ref{lemma grow L} and that both $\hat w$ and $\hat w_y$ satisfy PDE \eqref{PDE hatw}}. By adding up \eqref{PDE hatw} and $c\times$\eqref{equ pde wy}, it follows that
$$
\tilde{L}_t+\frac{\sigma^2}{2} \tilde{L}_{x x}=\frac{\sigma^2}{2 c} \frac{\hat{L}_x^2}{\hat{L}_y}+\frac{\sigma^2}{2} \frac{2 \hat{L}_x \hat{L}_y \hat{L}_{x y}-\hat{L}_x^2 \hat{L}_{y y}}{\hat{L}_y^2}=\hat{L}_t+\frac{\sigma^2}{2} \hat{L}_{x x},
$$
where the second equation is from \eqref{equ PDE hatL}. 
Consequently, we obtain 
\begin{equation}
(\hat{L}-\tilde{L})_t+\frac{\sigma^2}{2}(\hat{L}-\tilde{L})_{xx}=0, \quad \forall(x, y, t) \in \mathbb{R}^2 \times [0, T). 
\end{equation}
Furthermore, at time $T$ we have $$\tilde{L}(x, y, T) = \hat{w}(x, y, T)+c \hat{w}_y(x, y, T) = \hat{L}(x, y, T). $$
Since both $\hat L$ and $\tilde L$ grow at most exponentially,
by the comparison principle, we conclude that $\tilde{L}(x, y, t) = \hat{L}(x, y, t)$. That is, 
\[
\quad \hat{w}(x, y, t) + c \hat{w}_y(x, y, t) = \hat{L}(x, y, t).
\]

Finally, by the definition of $\hat{w}(x, y, t)$ (equation \eqref{PDE hatw}), we see that it indeed satisfies the PDE \eqref{finalpdew}. 

\vspace{8pt}
(ii). We next verify that $\frac{\hat{L}_x}{\hat{L}_y} \in [-\frac{c}{\sigma}c_\lambda, -\frac{c}{\sigma}\frac{1}{c_\lambda}]$ for some $c_\lambda > 1$. 
On the one hand, \eqref{equ LxLy sign} indicates that $\frac{\hat{L}_x}{\hat{L}_y}<0$. On the other hand, by \eqref{deri_YtoL}, we have $\frac{\hat{L}_x(x, y, t)}{\hat{L}_y(x, y, t)} = - \hat{Y}_x(x,\hat{L}(x, y, t),t)$ where $\hat{Y}$ is the unique classical solution of PDE \eqref{equ hat Y}. Notice that Lemma \ref{lem term Y} already verifies that $\hat{Y}_x^T$ is globally uniformly bounded for any $(x,l) \in \mathbb{R} \times \mathcal{L}$. Taking $\partial_x$ of PDE \eqref{equ mhat}, we see that $\hat {m}_x$ and $\hat m$ satisfy the same PDE operator.
Additionally, recalling that $\hat{Y}_x(x,l,t) = c \frac{\hat{m}_x(x,l,t)}{\hat{m}(x,l,t)}$, and it is bounded below and above at $t= T$ by Lemma \ref{lem term Y} (iii). The comparison principle then verifies the {uniform boundedness} of $\hat{Y}_x$.

\vspace{8pt}
(iii). From (i) and (ii), we see 
$\hat{w} \in C^{{5,2}}(\mathbb{R}^2 \times [0, T])$ satisfies the same PDE problem \eqref{equ w visco} as $w(x, y, t)$. 
Consequently, $\hat w = w$ follows from the following comparison principle.
\begin{lemma} \label{CP}
    (Comparison principle) Let $U$(resp. $V$) be an upper-semicontinuous viscosity sub-solution (resp. lower-semicontinuous viscosity super-solution) to the PDE \eqref{equ w visco}, such that $U(x,y,T) \leq V(x,y,T)$ on $\mathbb{R}^2$. Then $U \leq V$ on $\mathbb{R}^2\times [0,T]$.
\end{lemma}
The proof of this comparison principle is relegated to the appendix.
It follows from $w(x, y, t)  = \hat{w}(x, y, t)$ that ${w}(x, y, t) \in C^{{5,2}} (\mathbb{R}^2 \times[0, T])$, and $L_y= \hat{L}_y <0< \hat{L}_x = L_x$, $\frac{\hat{L}_x}{\hat{L}_y} = \frac{L_x}{L_y} \in [-\frac{c}{\sigma}c_\lambda, -\frac{c}{\sigma}\frac{1}{c_\lambda}]$.

\section{Conclusion}\label{conclusion}
In this paper, we prove the existence of a classical solution to a degenerate, fully nonlinear parabolic partial differential equation arising from the principal-agent problem. By establishing this result, we fill the gap in the literature and guarantee the existence of the optimal contract. The current work is restricted to a quadratic cost function, the extension to more general cases is left for future research.

\bibliographystyle{plain}

\clearpage

\appendix
\section*{Appendices}
\section{Proof of Proposition \ref{prop v C2}}
We consider the transformation $m(x, t): = \exp(v^C(x, t)/c)$. Then, \eqref{pde_vc} is equivalent to 
\begin{equation}\label{equ m0}
    \begin{cases}
        {m}_t+\frac{\sigma^2}{2} {m}_{x x}=0,\  \forall (x, t) \in \mathbb{R}\times[0, T),  \\
        {m}(x, T)=\exp\big(U_A(C(x))/c\big).\ \forall x \in \mathbb{R}.
    \end{cases}
\end{equation}
On the one hand, the existence and uniqueness of a classical solution to this PDE problem under the exponential growth condition can be derived from classical PDE theory (see 
\cite{ferretti2003uniqueness}, \cite{lieberman1996second},  \cite{pao2012nonlinear}). 
On the other hand, by the dynamic programming principle, $\exp(v^C(x, t)/c)$  should be a viscosity solution to the above PDE problem. By the comparison principle (see, e.g., \cite{dolcetta2007qualitative}), $\exp(v^C(x, t)/c)$ should be a classical solution to \eqref{equ m0}, which implies that $v^C(x, t)$ is a classical solution to \eqref{pde_vc} satisfying the growth condition \eqref{equ growth v}. 

The optimal effort follows from straightforward calculation. We complete the proof by showing that $\lambda^C$ and $\lambda^C_{x}$ are uniformly bounded. 
This result implies that
$\Lambda^C = \{\lambda^C(s)\}_{t \leq s \leq T} \in \mathcal{A}_\lambda(x,t)$ {(see \cite{karatzas2014brownian})}.

Indeed, $m_x(x,t)$ satisfying PDE \eqref{equ m0} with terminal condition 
\begin{align}
m_x(x,T) &= \frac{1}{c}\exp\big(U_A(C(x))/c\big) U_A'(C(x)) C'(x)\\ 
&\qquad \in [\frac{\exp\big(U_A(C(x))/c\big)}{c c_0c_1}, \frac{c_0 c_1 \exp\big(U_A(C(x))/c\big)}{c}].
\end{align}
Consequently, we have
$$
m_x(x,T) - \frac{m(x,T)}{c c_0 c_1} \geq 0, \quad \quad m_x(x,T) - \frac{c_0 c_1}{c} m(x,T) \leq 0.
$$
By the comparison principle, it follows that $\lambda^C(x,t) = \frac{\sigma}{c} v^C_x(x,t) = \sigma \frac{m_x(x,t)}{m(x,t)} \in [\frac{\sigma}{ c c_0 c_1},\frac{\sigma c_0 c_1}{c}]$, which is uniformly bounded.
{Following a similar argument, there exists a constant $c_m > 0$, {independent of $(x,t)$}, such that
$$
m_{xx}(x,T) + c_m m(x,T) \geq 0, \quad \quad m_{xx}(x,T) - c_m m(x,T) \leq 0.
$$
Thus, by comparison principle, $\lambda^C_x(x,t) = \frac{\sigma}{c} v^C_{xx}(x,t) = \sigma \frac{m_{xx}(x,t)}{m(x,t)} - \sigma \left(\frac{m_x(x,t)}{m(x,t)}\right)^2 \in [(-c_m-\frac{c_0^2c_1^2}{c^2})\sigma,(c_m - \frac{1}{c^2c_0^2c_1^2})\sigma]$, which is also uniformly bounded.}

\section{Proof of Proposition \ref{prop L}}
From \eqref{equ defi L}, we have
\begin{align*}
    d L(x^*(s), y^*(s), s) = d w(x^*(s), y^*(s), s) + c d w_y(x^*(s), y^*(s), s), \quad \forall s \in [t,T].
\end{align*}
On the one hand, applying It\^{o}'s lemma under the optimal control $\lambda^*$, we obtain
    \begin{align*}
    dw(x^*(s), y^*(s), s) &= w_{s} ds + w_{x} dx^*(s) + w_{y} dy^*(s) + \frac{1}{2} w_{xx} \langle dx^*(s)\rangle^2\\
    &\quad + \frac{1}{2} w_{yy} \langle dy^*(s)\rangle^2 + w_{xy}\langle dx^*(s), dy^*(s)\rangle\\
    &= w_{s} ds + w_{x} \big(\sigma \lambda^* d s+\sigma d B_s\big) + w_{y} \big(\frac{c}{2} (\lambda^*)^2 ds+c \lambda^* dB_s\big)\\
    & \quad + \frac{\sigma^2}{2} w_{xx} ds + \frac{(c \lambda^*)^2}{2} w_{yy} ds + w_{xy} \sigma c \lambda^* ds\\
    &= \bigg( w_{s} + \frac{\sigma^2}{2} w_{xx} + \frac{(c \lambda^*)^2}{2} w_{yy} + \sigma\lambda^* w_{x}\\
    &\qquad + \frac{c}{2} (\lambda^*)^2 w_{y} + \sigma c \lambda^* w_{xy}\bigg) ds + \big( \sigma w_{x} + c \lambda^* w_{y}\big)dB_{s}\\
    &= \big( \sigma w_{x} + c \lambda^* w_{y}\big)dB_{s},
    \end{align*}
where the last equation follows from PDE \eqref{equ w visco}.

On the other hand, we have
\begin{align*}
    dw_y(x^*(s), y^*(s), s) &= w_{ys} ds + w_{yx} dx^*(s) + w_{yy} dy^*(s) + \frac{1}{2} w_{yxx} \langle dx^*(s)\rangle^2\\
    &\quad + \frac{1}{2} w_{yyy} \langle dy^*(s)\rangle^2 + w_{xyy}\langle dx^*(s), dy^*(s)\rangle\\
    &= w_{ys} ds + w_{yx} \big(\sigma \lambda^* d s+\sigma d B_s\big) + w_{yy} \big(\frac{c}{2} (\lambda^*)^2 ds+c \lambda^* dB_s\big)\\
    & \quad + \frac{\sigma^2}{2} w_{yxx} ds + \frac{(c \lambda^*)^2}{2} w_{yyy} ds + w_{xyy} \sigma c \lambda^* ds\\
    &= \bigg( w_{ys} + \frac{\sigma^2}{2} w_{yxx} + \frac{(c \lambda^*)^2}{2} w_{yyy} + \sigma \lambda^* w_{yx}\\
    &\qquad + \frac{c}{2} (\lambda^*)^2 w_{yy} + \sigma c \lambda^* w_{xyy}\bigg) ds + \big( \sigma w_{yx} + c \lambda^* w_{yy}\big)dB_{s}\\
    &= \big( \sigma w_{yx} + c \lambda^* w_{yy}\big)dB_{s},
\end{align*}
where the last equation follows from taking $\frac{\partial}{\partial y}$ on both sides of \eqref{equ w visco}.

Finally, combining the results above and utilizing \eqref{def alpha_star}, we obtain
\begin{align*}
    d L(x^*(s), y^*(s), s) &= d w(x^*(s), y^*(s), s) + c d w_y(x^*(s), y^*(s), s) \\
    &= \big(\sigma w_{x} + c \lambda^* w_{y} + c\sigma w_{yx} + c^2 \lambda^* w_{yy}\big)dB_{s}\\
    &= 0,
\end{align*}
which follows from the first-order condition of \eqref{equ w visco}.
Therefore, $$L(x^*(s), y^*(s), s) = L(x, y, t), \quad \forall s \in [t, T].$$
The terminal condition follows immediately from $w(x, y, T) = U_P \big( x - U_A^{-1}(y) \big)$.

\section{Proof of Theorem \ref{thm equi}}
First, we prove that $u(x, y, t) \leqslant w(x, y, t)$, $\forall x, y \in \mathbb{R}, t\in [0, T]$.
For any $C \in \mathcal{C}(x, y, t)$, it follows from the definition that 
\[
\begin{aligned}
w(x, y, t) &= \sup_{\Lambda \in \mathcal{A}(x, y, t)} \mathbb{E} \left[ U_P \bigg( x^{\lambda}(T)  - U_A^{-1} \left( y^{\lambda}(T) \right) \bigg)\right] \\
&\geqslant \mathbb{E} \left[ U_P \bigg( x^C(T) - U_A^{-1} \left( y^C(T) \right) \bigg) \right] \\
&= \mathbb{E} \left[ U_P \left( x^C(T) - C(x^C(T)\right) \right].
\end{aligned}
\]
Consequently, $w(x, y, t) \geq \sup \limits_{C \in \mathcal{C}(x, y, t)} \mathbb{E}\left[U_P\left(x^C(T)-C\left(x^C(T)\right)\right]=u(x, y, t)\right.$.

 Next, we show that $u(x, y, t) \geqslant w(x, y, t)$, $\forall x, y \in \mathbb{R}, t\in [0, T]$. From Proposition \ref{prop L}, given a constant $l$  
 we can define a function $v^*(x, s) $ as the real value satisfying  
 \begin{align}\label{equ v*}
   l = L(x, v^*(x, s), s).  
  \end{align}
 {In particular, we define $C^*(x; l)$ by \eqref{equ C*0}. It then follows from \eqref{equ v*} that }
 \begin{align}\label{equ C*}
 C^*(x; l) = U_A^{-1}(v^*(x, T)). 
 \end{align}

The following two lemmas justify that $v^*$ is well-defined.  

We begin by considering the terminal time $T$. 
\begin{lemma} \label{prop v* exit}
Assume {$l \in \mathcal{L}$.}
Then for any $x\in \mathbb{R}$, $v^*(x, T)$ is well-defined, which is $C^6$ and grows at most linearly in $x$.  
\end{lemma} 
 \begin{proof}[Proof of Lemma \ref{prop v* exit}]
 The existence, uniqueness, {regularity,}
 and growth condition of $v^*(x,T)$ follow immediately from Lemma \ref{lem term Y} (i).  
 \end{proof}
We now consider the general time $t<T$.  
 \begin{lemma} \label{prop_pde_v_star}
     The solution $v^*(x, t)$ to \eqref{equ v*} exists for any $x\in \mathbb{R}$, $t\in [0, T]$. Moreover, it is the unique classical solution to
\begin{equation}
    \begin{cases}
        v_s(x, s)+\frac{\sigma^2}{2} v_{xx}(x, s)+\frac{\sigma^2}{2 c}(v_x(x, s))^2=0 \quad \forall x \in \mathbb{R}, s \in[t, T), \\
        v(x, T)=U_A(C^*(x; l)) \quad \forall x \in \mathbb{R}. 
    \end{cases}
    \label{pde_v_star}
\end{equation}
 \end{lemma}
\begin{proof}[Proof of Lemma \ref{prop_pde_v_star}]
The above PDE admits a unique classical solution $v^{C^*}(x, s)$ satisfying a linear growth condition. This follows from the fact that $m(x, t): = \exp(v^{C^*}(x, t)/c)$ satisfies the linear PDE \eqref{equ m0}, along with {\eqref{equ C*} and the linear growth condition of $v^*(x,T)$ established in Lemma \ref{prop v* exit}}. 
 
  Define the domain $$\Omega_v:= \{(x, s) \in \mathbb{R}\times[0, T]| |v^*(x, s)|< \infty \},$$ which is a relative open subset of $\mathbb{R}\times [0, T]$ containing $\mathbb{R}\times \{T\}$. Within $\Omega_v$, implicit differentiation of \( l = L(x, v^*(z, s), s) \) yields
    \begin{align}
    0 &= \frac{\partial}{\partial s} L\left(x, v^*(x, s), s\right) = L_s+ L_v v_s^*, \label{equ Ls}\\
    0 &= \frac{\partial }{\partial x} L\left(x, v^*(x, s), s\right) = L_x + L_v v_x^*, \label{equ Lz}\\
    0 &= \frac{\partial^2}{\partial x^2} L\left(x, v^*(x, s), s\right) = L_{xx} + 2 L_{xv} v_x^*+ L_{vv}\left(v_x^*\right)^2 + L_v v_{xx}^*. \label{equ Lzz}
    \end{align}
    From \eqref{equ Lz}, we have
    \begin{equation}
    v_x^*(x, s) = -\frac{L_x}{L_v}.
    \label{v_x_star}
    \end{equation}
    
     Furthermore, from \eqref{operator_w}, we have 
    \begin{equation}
    w_s(x, v, s) + \frac{\sigma^2}{2} w_{xx}(x, v, s) = \frac{\sigma^2}{2 c} \frac{L_x^2(x, v, s)}{L_v(x, v, s)} \quad \text{for } x \in \mathbb{R}, v \in \mathbb{R}, s \in [0, T]. 
    \label{w&L}
    \end{equation}
   Differentiating the equation with respect to \( v \), we obtain
    \begin{equation}
    c\left(w_v\right)_s + \frac{c  \sigma^2}{2}\left(w_v\right)_{xx} = \sigma^2 \frac{L_x L_{xv}}{L_v} - \frac{\sigma^2 L_x^2 L_{vv}}{2 L_v^2}.
    \label{w&L_v}
    \end{equation}
    By the definition \( L(x, v, s) = w(x, v, s) + c w_v(x, v, s) \), adding \eqref{w&L} and \eqref{w&L_v} yields
    \begin{equation}
    L_s + \frac{\sigma^2}{2} L_{xx} = \frac{\sigma^2}{2 c} \frac{L_x^2}{L_v} + \sigma^2 \frac{L_x L_{xv}}{L_v} - \frac{\sigma^2 L_x^2 L_{vv}}{2 L_v^2}.
    \label{equ_L}
    \end{equation}
    Next, adding \eqref{equ Ls} and \eqref{equ Lzz} $\times \frac{\sigma^2}{2} $, we obtain
    \[
    0 = L_s + L_v v_s^* + \frac{\sigma^2}{2}\left(L_{xx} + 2 L_{xv} v_x^* + L_{vv} v_x^{*2} + L_v v_{xx}^*\right).
    \]
    
	Substitute \eqref{v_x_star} and \eqref{equ_L} into the equation above, we obtain
    \begin{equation}
    \begin{aligned}
        0 &= L_v v_s^* + \frac{\sigma^2}{2c} \frac{L_x^2}{L_v} + \frac{\sigma^2 L_x L_{xv}}{L_v} - \frac{\sigma^2 L_x^2 L_{vv}}{2L_v^2} + \frac{\sigma^2}{2}\left(2 L_{xv} v_x^* + L_{vv} {v_{x}^*}^2 + L_v v_{xx}^*\right)\\
        & = L_v v_s^* + \frac{\sigma^2}{2c} \frac{L_x^2}{L_v} + \frac{\sigma^2}{2} L_v v_{xx}^*\\
        &= {L_v}\left(v_s^* + \frac{\sigma^2}{2} v_{xx}^* + \frac{\sigma^2}{2c}\left(v_x^*\right)^2\right).
        \notag
    \end{aligned}
    \end{equation}
Since $L_v <0$ by assumption, it follows that $v^*$ satisfies \eqref{pde_v_star} in $\Omega_v$, with the same terminal condition. This result, combined with the comparison principle, proves the claim. 
\end{proof}

\begin{proposition} \label{prop_C_prime_bdd}
  For any fixed $l$,
  {$C^*(x; l)\in \mathcal{C}(x,y,t)$.}
\end{proposition}
\begin{proof}[Proof of Proposition \ref{prop_C_prime_bdd}]
    First, the {regularity} and growth condition of $C^*$ are verified by Lemma \ref{prop v* exit}. 
    
    We next show $\frac{1}{c_1} \leq (C^*)'(x;l) \leq c_1$ {and $-c_1 \leq (C^*)''(x;l) \leq c_1$} for some constant $c_1 > 1$. Indeed, by \eqref{equ C*} and \eqref{v_x_star}, we have
    \begin{align}
        (C^*)'(x;l) &= \frac{1}{U_A'(U_A^{-1}(v^*(x, T)))} v^*_x(x,T) = -\frac{1}{U_A'(U_A^{-1}(v^*(x, T)))} \frac{L_x\big(x,v^*(x, T),T\big)}{L_v\big(x,v^*(x, T),T\big)},
    \end{align}
    {
    and
    \begin{align}
        (C^*)''(x;l) &= \frac{v^*_{xx} U_A' - (v^*_x)^2 U_A'' \frac{1}{U_A'}}{(U_A')^2} = \frac{v^*_{xx}}{U_A'} - \frac{(v^*_x)^2 U_A''}{(U_A')^3}.
    \end{align}
    }
    {It follows immediately from Lemma \ref{lem term Y} (iii) that $v^*_{x}(x,T)$ and $v^*_{xx}(x,T)$ are both uniformly bounded. Consequently, the proof is completed by Assumption \ref{assumption Ua}. It follows that $C^*(x; l)\in\mathscr{C}$.}
    
    Finally, we verify the existence of strong solutions and \eqref{contract0}.
    By Proposition \ref{prop v C2}, we have $\lambda^{C^*} \in \mathcal{A}_\lambda(x,t)$. Furthermore, we can define $y^{C^*}(s)$ by \eqref{dyn_y_c} so that the strong solutions exist. The inequality \eqref{contract0} follows from $l = L(x,y,t) = L(x,v^*(x,t),t)$ and $L_y < 0$.
\end{proof}

Therefore, for any $(x, y, t)$, we can calculate $l=L(x, y, t)$ and define the contract as in \eqref{equ C*}. 
By the equality $L(x^*(s), y^*(s), s) = l = L(x^*(s), v^*(x^*(s), s), s)$ from Proposition \ref{prop L} and the definition of $v^*(x^*(s), s)$, and the fact that $L_y<0$, it follows that $v^*(x^*(s), s) = y^*(s)$.
Consequently, together with \eqref{alpha_c} and  
\eqref{v_x_star}, we obtain 
$$
\begin{aligned}
{\lambda}^*(s) & \left.=-\frac{\sigma}{c} \frac{L_x\left(x^*(s), y^*(s), s\right)}{L_y\left(x^*(s), y^*(s), s\right)} \right| _{y^*(s)=v^*\left(x^*(s), s\right)} \\
& =\frac{\sigma}{c} v^*_x\left(x^*(s), s\right)={\lambda}^{C^*}(s) .
\end{aligned}
$$
Thus, by Proposition \ref{prop_C_prime_bdd} and \eqref{dyn_y_c}, $\Lambda^* = \{\lambda^*(s)\}_{t \leq s \leq T} \in \mathcal{A}(x,y,t)$ and 
\begin{align}\label{equ w to u 1}
    w(x, y, t) & =\mathbb{E}\big[U_P(x^*(T)-U_A^{-1}(y^*(T))\big]. 
\end{align}
Finally, it follows from \eqref{equ C*} that the right hand side of \eqref{equ w to u 1} satisfies
\begin{align}
    \mathbb{E}\big[U_P(x^*(T)-U_A^{-1}(y^*(T))\big] = &
    \mathbb{E}\left[U_P\left(x^{C^*}(T)-U_A^{-1}\big(v^*(x^*(T), T)\big)\right)\right]\\
    = & 
    \mathbb{E}\left[U_P\left(x^{C^*}(T)-C^*(x^{C^*}(T))\right)\right]\\
    \leq &  \sup_{C \in \mathcal{C}(x, y, t)} \mathbb{E}\left[ U_P\left( x^C(T) \right) - C\left( x^C(T) \right) \right] \\
    = & u(x, y, t).
\end{align}
This completes the proof.

\section{Proof of Lemma \ref{lem term Y}}
  {\bf Proof of $(ii)$:} Recall that
    \begin{equation}
     L(x, y, T)=w(x, y, T)+cw_y(x, y, T) = U_P\left(x-U_A^{-1}(y)\right)-c \frac{U_P^{\prime}\left(x-U_A^{-1}(y)\right) }{U_A^{\prime}(U_A^{-1}(y))},
    \end{equation}
    by the Assumption \ref{assumption Ua} and \ref{assumption Up},
    we finish the proof of $(ii)$ with
    \begin{align}\label{equ l UP}
      l &= U_P\left(x-U_A^{-1}(\hat{Y}^T)\right)-c U_P^{\prime}\left(x-U_A^{-1}(\hat{Y}^T)\right) \frac{1}{U_A'(U_A^{-1}(\hat{Y}^T))}\\
      &\sim  U_P\left(x-U_A^{-1}(\hat{Y}^T)\right),
    \end{align}
since $U'_A$ is bounded and Assumption \ref{assumption Up} implies that $U'_P$ and $U_P$ are of the same order with opposite signs. 
    
    \vspace{8pt}
    {\bf Proof of $(iii)$:} We first show that $\hat Y^T_x$ is uniformly bounded. 
    Taking $\partial_x$ on both sides of \eqref{equ term hatY} and we obtain 
    $$ 0 = L_x + L_y \hat{Y}_x^T. $$
    Given the growth rate of $\hat{Y}^T$ from Lemma \ref{lem term Y} (i), together with Assumption \ref{assumption Ua} and Assumption \ref{assumption Up}, it follows that
    \begin{align}
    L_x &= U'_P-c \frac{U_P^{\prime\prime}}{{U_A^{\prime}}} {\sim l},\\
    L_y &= -\frac{U_P'}{U_A'} + c \frac{U_P''}{(U_A')^2} + c \frac{U_P' U_A''}{(U_A')^3} {\sim l}.\label{equ Ly sim l}
    \end{align}
    Moreover, 
    \begin{align}
        \frac{L_y}{L_x} &= \frac{-\frac{U_P'}{U_A'} + c \frac{U_P''}{(U_A')^2}}{U'_P-c \frac{U_P^{\prime\prime}}{U_A'}} + \frac{c U_P' \frac{U_A''}{(U_A')^3}}{U'_P-c \frac{U_P^{\prime\prime}}{U_A'}} = -\frac{1}{U_A'} + \frac{c U_P' \frac{U_A''}{(U_A')^3}}{U'_P-c \frac{U_P^{\prime\prime}}{U_A'}}.
    \end{align}
    Noticing that
    $$
    0 \geq \frac{c U_P' \frac{U_A''}{(U_A')^3}}{U'_P-c \frac{U_P^{\prime\prime}}{U_A'}} \geq \frac{c U_P' \frac{U_A''}{(U_A')^3}}{U'_P} = c \frac{U_A''}{(U_A')^3} \geq -c \cdot c_0^4,
    $$
    we see $\hat{Y}_x^T \in [\frac{1}{c_0 + c \cdot c_0^4}, c_0]$, which implies that it is globally uniformly bounded.

    We next show that $\hat{Y}_{xx}^T$ is globally uniformly bounded. We have
    \begin{align}
        \hat{Y}^T_{xx}(x,l) &= \frac{2L_x L_y L_{xy}-L_{xx}L_y^2-L_{yy}L_x^2}{(L_y)^3} = 2\frac{L_x}{L_y}\frac{L_{xy}}{L_y} - \frac{L_{xx}}{L_y} - \frac{L_{yy}}{L_y} \left(\frac{L_x}{L_y}\right)^2.
    \end{align}
    {Since $\hat{Y}_x^T = -\frac{L_x}{L_y}$ is globally uniformly bounded, it is sufficient to show that $\frac{L_{xy}}{L_{y}}$, $\frac{L_{yy}}{L_{y}}$, and $\frac{L_{xx}}{L_{y}}$ are bounded.}

    By direct calculation, we have
    \begin{align}
         L_y &= -\frac{U_P'}{U_A'} + c \frac{U_P''}{(U_A')^2} + c \frac{U_P' U_A''}{(U_A')^3}\leq -\frac{U_P'}{U_A'} <0, \\
        L_{xx} &= U_P'' - c \frac{U^{(3)}_P}{U_A'},\\
        L_{xy} &= -\frac{U_P''}{U_A'} - c \frac{-U_P^{(3)} \frac{1}{U_A'} U_A' - U_A''\frac{1}{U_A'}U_P''}{(U_A')^2} = - \frac{U_P''}{U_A'} + c\frac{U_P^{(3)}}{(U_A')^2} + c \frac{U_A'' U_P''}{(U_A')^3},\\
        L_{yy} 
        &= \frac{U_P''}{(U_A')^2} + \frac{U_A''U_P'}{(U_A')^3} - \frac{cU_P^{(3)}}{(U_A')^3} - \frac{ 3c U_A'' U_P''}{(U_A')^4} + \frac{c U_P' U_A^{(3)}}{(U_A')^4} - \frac{3 c (U_A'')^2 U_P'}{(U_A')^5},.
    \end{align}
    The proof is completed by noticing that $U'_A$, $U''_A$ are bounded, and that $U''_P$ and $U^{(3)}_P$ are controlled by $U'_P$ under Assumption \ref{assumption Ua} and Assumption \ref{assumption Up}.
    
    {Finally we prove that $\hat{Y}_l^T \sim \frac{1}{l}$. Taking $\partial_l$ on both sides of \eqref{equ term hatY}, we obtain
        $$ 1 = L_y \hat{Y}_l^T. $$
    By \eqref{equ Ly sim l}, it follows that
    \begin{align}
        \hat{Y}_l^T \sim \frac{1}{l} \sim \frac{1}{U_P\left(x-U_A^{-1}(\hat{Y}^T)\right)}.
    \end{align}
    }
    
    \vspace{8pt}
    {{\bf Proof of $(iv)$:} We proceed by induction. For integers $2 \leq k \leq 6$, we observe that
    $$
    \frac{\partial^k \hat{Y}^T}{\partial x^k} = \frac{\frac{\partial^k L}{\partial x^k} + \frac{\partial^k L}{\partial y^k} ( \hat{Y}_x^T)^k + k \sum_{i=1}^{k-1}\frac{\partial^{k} L}{\partial x^{k-i} \partial y^{i}}(\hat{Y}_x^T)^i + \text{lower order terms}}{-\frac{\partial L}{\partial y}},
    $$
    Indeed, each $\frac{\partial^{k} L}{\partial x^{k-i} \partial y^{i}}, \, 0 \leq i \leq k$, $2 \leq k \leq 6$, is a sum of finitely many terms, each of which is a ratio of $U_P^{(m)}U_A^{(n)}$ to $(U_A')^{p}$.
    Notice that $U_A^{(n)}$ and $(U_A')^{p}$ are all bounded from Assumption \ref{assumption Ua}.} Thus, given \eqref{equ UP derivartive}, we have 
    $$
    \left|\frac{\partial^k \hat{Y}^T}{\partial x^k}\right| \leq \frac{g(|l|)}{|l|^d},
    $$
    where $d > 0$ is a finite constant and $g(\cdot)$ is a polynomial function of finite degree. Although $g$ and $d$ depend on $k$, the finiteness of $k$ allows us to take a function $g$ and constant $d$ valid for all $2\leq k\leq 6$. 
    {Applying the same argument to $\frac{\partial^{k} \hat{Y}^T(x, l)}{\partial l^k}$ and $\frac{\partial^{i+j} \hat{Y}^T(x, l)}{\partial x^i \partial l^j}$ completes the proof of $(iv)$.
    }

\section{Proof of Lemma \ref{lemma grow L}}

Before everything, we prove the following lemma for the later use. 
\begin{lemma} \label{lemma_growth_derivative_YT}
{$\left|\hat{Y}^T(x,l) - \hat{Y}(x,l,t)\right| \leq c_y$ for some constant $c_y \geq 0$ independent of $(x, l)$ and $t$.}
\end{lemma}
\begin{proof}[Proof of  Lemma \ref{lemma_growth_derivative_YT}]
From Lemma \ref{lem term Y} (iii), we have $\hat{Y}^T_x$ and $\hat{Y}^T_{xx}$ are both globally uniformly bounded. Thus, we consider $\hat{Z}(x,l,t) = \hat{Y}^T(x,l) + c_y(T-t)$, which is a viscosity super-solution to PDE \eqref{equ hat Y}. Also, {$e^{\frac{\hat{Z}(x,l,t)}{c}}$ is a viscosity super-solution} to PDE \eqref{equ mhat} for a sufficiently large $c_y$. By the comparison principle for PDE \eqref{equ mhat}, we have $\hat{Y}^T(x,l) + c_y(T-t) \geq \hat{Y}(x,l,t)$. Similarly, $\hat{Y}^T(x,l) - c_y(T-t) \leq \hat{Y}(x,l,t)$. Hence, we complete the proof.
\end{proof}

We prove Lemma \ref{lemma grow L} in six steps.

\vspace{8pt}
{{$(i)$:} We show that $\left|\hat{L}\right|\leq c_2 ( 1 + e^{c_2 ({|x|+|y|})})$.}

By Lemma \ref{lem term Y} (ii), we have
$$
0>\hat{L} \sim U_P\left(x-U_A^{-1}(\hat{Y}^T(x,\hat{L}))\right).
$$
Moreover, by Assumption \ref{assumption Up},
we have  
\begin{align}
 |U_P\left(x-U_A^{-1}(\hat{Y}^T(x, \hat{L}))\right)| 
&\leq  c_0 \big(1+ e^{c_0 |x-U_A^{-1}(\hat{Y}^T(x, \hat{L})) |} \big)\\
&\leq  c_0 \big(1+ e^{c_0 (|x| + |U_A^{-1}(\hat{Y}^T(x, \hat{L}))|)} \big).
\end{align}
Additionally, by Lemma \ref{lemma_growth_derivative_YT}, we have
\begin{align}
|U_A^{-1}(\hat{Y}^T(x, \hat{L}))| &\leq |U_A^{-1}(\hat{Y}(x, \hat{L}, t))| + \max_{z} |(U_A^{-1})'(z)|  c_y \\
&= |U_A^{-1}(y)| + \max_{z} |(U_A^{-1})'(z)| c_y\\
&\leq c_0' (1+ |y|)  
\end{align}
for a constant $c_0'>0$.
Therefore,
\begin{align}
    |\hat{L}| \lesssim c_2 ( 1 + e^{c_2 ({|x|+|y|})}).
\end{align}

\vspace{8pt}
{$(ii)$:} We consider $\frac{(\hat{L}_x)^2}{\hat{L}_y}$.

From equation \eqref{deri_YtoL} we see 
    \begin{align}\label{equ L Y relation}
    \frac{(\hat{L}_x(x, y, t))^2}{\hat{L}_y(x, y, t)} = \frac{(\hat{Y}_x(x, l, t))^2}{\hat{Y}_l(x, l, t)}.
    \end{align}
 
{On the one hand, Lemma \ref{lem term Y} already verifies that $\hat{Y}_x^T$ is globally uniformly bounded for any $(x,l) \in \mathbb{R} \times \mathcal{L}$. By taking $\partial_x$ on PDE \eqref{equ mhat}, we see that $\hat {m}_x$ and $\hat m$ satisfy the same PDE operator.
Since $\hat{Y}_x(x,l,t) = c \frac{\hat{m}_x(x,l,t)}{\hat{m}(x,l,t)}$ is bounded at $t= T$ by Lemma \ref{lem term Y} (iii), the comparison principle ensures the {uniform boundedness} of $\hat{Y}_x$.
}
On the other hand, by taking $\partial_l$ on PDE \eqref{equ mhat}, we see that $\hat {m}_l$ and $\hat m$ also satisfy the same PDE operator. Given that $\hat{Y}_l(x,l,t) = c \frac{\hat{m}_l(x,l,t)}{\hat{m}(x,l,t)}$ and $\hat{Y}^T_l(x,l) \sim \frac{1}{l}$ from Lemma \ref{lem term Y} (iii), the comparison principle for PDE \eqref{equ mhat} yields
    \begin{align}\label{equ Yl es 1}
        \hat{Y}_l(x,l,t) \sim \frac{1}{l} \sim \frac{1}{U_P\left(x-U_A^{-1}(\hat{Y}^T(x, l))\right)}.
    \end{align}
By Assumption \ref{assumption Up},  
we have  
\begin{align}
 |U_P\left(x-U_A^{-1}(\hat{Y}^T(x, l))\right)| 
&\leq  c_0 \big(1+ e^{c_0 |x-U_A^{-1}(\hat{Y}^T(x, l)) |} \big)\\
&\leq  c_0 \big(1+ e^{c_0 (|x| + |U_A^{-1}(\hat{Y}^T(x, l))|)} \big).
\end{align}
Moreover, by Lemma \ref{lemma_growth_derivative_YT}, we have
\begin{align}
|U_A^{-1}(\hat{Y}^T(x, l))| &\leq |U_A^{-1}(\hat{Y}(x, l, t))| + \max_{z} |(U_A^{-1})'(z)|  c_y \\
&= |U_A^{-1}(y)| + \max_{z} |(U_A^{-1})'(z)| c_y\\
&\leq c_0' (1+ |y|)  
\end{align}
for a constant $c_0'>0$. 

Therefore,
\begin{align}
    |\frac{(\hat{L}_x(x, y, t))^2}{\hat{L}_y(x, y, t)}|  = |\frac{(\hat{Y}_x(x, l, t))^2}{\hat{Y}_l(x, l, t)}| \lesssim c_2 ( 1 + e^{c_2 ({|x|+|y|})}).  \label{equ Yl es 2}
\end{align}

\vspace{8pt}
{$(iii)$:} We show that $\left|\frac{\partial^{i+j} \left(\frac{(\hat{L}_x)^2}{\hat{L}_y}\right)}{\partial x^i \partial y^j}\right| \leq c_2 ( 1 + e^{c_2 ({|x|+|y|})})$. 

Notice that
\begin{align}
    \frac{\partial^{i+j} \left(\frac{(\hat{L}_x)^2}{\hat{L}_y}\right)}{\partial x^i \partial y^j} &= \frac{\sum_{k=0}^{i+j+1}(\text{Lower order terms})\cdot\frac{\partial^{i+j+1} \hat{L}}{\partial x^{k}\partial y^{i+j+1-k}} + \text{Lower order terms}}{(\hat{L}_y)^m}.
\end{align}
Since we already show the growth condition of $\hat{L}_x$ and $\hat{L}_y$, it is sufficient to show that $\frac{\partial^{p} \hat{L}}{\partial x^{i} \partial y^{p-i}}, \, 0 \leq i \leq p$ has at most exponential growth for each $2 \leq p \leq 6$. Indeed, by taking partial derivatives $\partial_x$, and $\partial_y$ in \eqref{equ defi hatL}, we have
\begin{align}
\begin{cases}
0=\hat{Y}_x+\hat{Y}_l \hat{L}_x, \\
1=\hat{Y}_l \hat{L}_y.
\end{cases}
\end{align}
This indicates that we can express $\frac{\partial^{p} \hat{L}}{\partial x^{i} \partial y^{p-i}}$ in terms of the partial derivatives of $\hat{Y}$, up to the order $p$. 
Similar to (ii), the comparison principle and Lemma \ref{lem term Y} (iii) and (iv) imply that these derivatives of $\hat{Y}$ are no more than $g(|l|)/|l|^d$, in terms of the absolute value. Then similar to \eqref{equ Yl es 1}-\eqref{equ Yl es 2}, we derive (iii).

\vspace{8pt}
{$(iv)$:} We show that ${\left|\frac{\partial \left(\frac{(\hat{L}_x)^2}{\hat{L}_y}\right)}{\partial t}\right|} \leq c_2 ( 1 + e^{c_2 ({|x|+|y|})})$. 

We observe that
\begin{align}
    \frac{\partial \left(\frac{(\hat{L}_x)^2}{\hat{L}_y}\right)}{\partial t} &= \frac{2\hat{L}_x \hat{L}_{xt} \hat{L}_y-\hat{L}_{yt} (\hat{L}_x)^2}{(\hat{L}_y)^2}.
\end{align}
From \eqref{equ defi hatL}, by taking partial derivatives {$\partial_{t}$}, $\partial_{xt}$ and {$\partial_{yt}$}, we have
\begin{align}
    \begin{cases}
   0 &= {\hat{Y}_t+\hat{Y}_l \hat{L}_t,} \\
    0 &= \hat{Y}_{tx} + \hat{Y}_{tl} \hat{L}_{x} + \hat{Y}_{lx} \hat{L}_{t} + \hat{Y}_{ll} \hat{L}_{t} \hat{L}_{x} + \hat{Y}_{l} \hat{L}_{xt},\\
    0 &= \hat{Y}_{lt} \hat{L}_{y} + \hat{Y}_{ll}  \hat{L}_{t} \hat{L}_{y} + \hat{Y}_{l} \hat{L}_{yt}.
    \end{cases}
\end{align}
Therefore, it suffices to analyze $\hat{Y}_t$, $\hat{Y}_{tx}$, and $\hat{Y}_{tl}$.
Since $\hat{Y}$ satisfies \eqref{equ hat Y}, it follows from Lemma \ref{lem term Y} that $|\hat{Y}_t|\leq \frac{G(|l|)}{|l|^D}$ for some polynomial $G$ and integer $D>0$. 

\vspace{8pt}
{$(v)$:} We show that $\left|\frac{\partial^2 \left(\frac{(\hat{L}_x)^2}{\hat{L}_y}\right)}{\partial t^2}\right| \leq c_2 ( 1 + e^{c_2 ({|x|+|y|})})$. 

As in (iv), from \eqref{equ defi hatL}, we only need to analyze $\hat{Y}_{tt}$, $\hat{Y}_{txt}$, $\hat{Y}_{txl}$, $\hat{Y}_{llt}$, and $\hat{Y}_{ltt}$. The same arguments work here. 

\vspace{8pt}
{{$(vi)$:} We show that$\left|\frac{\partial^2 \left(\frac{(\hat{L}_x)^2}{\hat{L}_y}\right)}{\partial t \partial x}\right|,\ \left|\frac{\partial^2 \left(\frac{(\hat{L}_x)^2}{\hat{L}_y}\right)}{\partial t \partial y}\right| \leq c_2 ( 1 + e^{c_2 ({|x|+|y|})})$.}

{
We observe that
\begin{align}
    \frac{\partial^2 \left(\frac{(\hat{L}_x)^2}{\hat{L}_y}\right)}{\partial t \partial x} &= \frac{\partial}{\partial x}\left(\frac{2\hat{L}_x \hat{L}_{xt} \hat{L}_y-\hat{L}_{yt} (\hat{L}_x)^2}{(\hat{L}_y)^2}\right),\\
     \frac{\partial^2 \left(\frac{(\hat{L}_x)^2}{\hat{L}_y}\right)}{\partial t \partial y} &=  \frac{\partial}{\partial y}\left(\frac{2\hat{L}_x \hat{L}_{xt} \hat{L}_y-\hat{L}_{yt} (\hat{L}_x)^2}{(\hat{L}_y)^2}\right).
\end{align}
Like (iv), from \eqref{equ defi hatL}, we only need to analyze $\hat{Y}_{txx}$ and the same arguments hold.
}

\section{Proof of Lemma \ref{CP} (Comparison principle)}
    \textbf{Step 1:} Let $\tilde{U} = e^{\beta t}U$ and $\tilde{V} = e^{\beta t}V$. By direct calculation, $\tilde{U}$(resp. $\tilde{V}$) is a viscosity sub-solution(resp. viscosity super-solution) to
    \begin{equation} \label{pde_beta_CP}
        -w_t + \beta w - \sup \limits_{{\lambda} \in [\frac{1}{c_\lambda},c_\lambda]} \bigg\{ \frac{\sigma^2}{2} w_{xx} + \sigma c {\lambda} w_{xy} + \frac{c^2 {\lambda}^2}{2} w_{yy} +  \sigma {\lambda} w_x + \frac{c}{2} {\lambda}^2 w_y \bigg\} = 0.
    \end{equation}
    Hence, without loss of generality, we assume $\beta > 0$ and focus on the PDE \eqref{pde_beta_CP}.

     Let $\mu > 0$ and define $U_\mu = U - \frac{\mu}{t}$. It is straightforward to verify that $- \frac{\mu}{t}$ is a sub-solution to the PDE \eqref{pde_beta_CP}. Consequently, $U_\mu$ is a viscosity sub-solution to the PDE \eqref{pde_beta_CP} satisfying the growth condition specified in \eqref{equ w visco}. Additionally, $U_\mu$ is upper semi-continuous for $t > 0$, and satisfies the following condition:
    \begin{equation} \label{CP_initial_condition}
        \lim_{t \searrow 0} \sup_{x,y \in \mathbb{R}} U_\mu(x,y,t) = -\infty.
    \end{equation}
    To establish $U \leq V$, it is sufficient to show that $U_\mu \leq V$ for any $\mu>0$. 

    \textbf{Step 2:} We consider
    $$
    \phi(x,y) := e^{2D (x+y)} + e^{2D (x-y)} + e^{2D (-x+y)} + e^{2D (-x-y)}.
    $$
    Notice that $\phi_x = 2D e^{2D (x+y)} + 2D e^{2D (x-y)} - 2D e^{2D (-x+y)} - 2De^{2D (-x-y)} \leq 2D \cdot \phi(x,y)$. Similarly, it is straightforward to verify that $\phi_{y} \leq 2D \cdot \phi(x,y)$, $\phi_{xx} \leq 4D^2 \cdot \phi(x,y)$, $\phi_{xy} \leq 4D^2 \cdot \phi(x,y)$ and $\phi_{yy} \leq 4D^2 \cdot \phi(x,y)$.
    Consequently, given the boundedness of $\lambda$, there exists a constant $d > 0$ such that
    \begin{align*}
        &-\phi_t + \beta \phi - \sup \limits_{{\lambda} \in [\frac{1}{c_\lambda},c_\lambda]} \bigg\{ \frac{\sigma^2}{2} \phi_{xx} + \sigma c {\lambda} \phi_{xy} + \frac{c^2 {\lambda}^2}{2} \phi_{yy} +  \sigma {\lambda} \phi_x + \frac{c}{2} {\lambda}^2 \phi_y \bigg\},\\
        &= \beta \phi - \sup \limits_{{\lambda} \in [\frac{1}{c_\lambda},c_\lambda]} \bigg\{ \frac{\sigma^2}{2} \phi_{xx} + \sigma c {\lambda} \phi_{xy} + \frac{c^2 {\lambda}^2}{2} \phi_{yy} +  \sigma {\lambda} \phi_x + \frac{c}{2} {\lambda}^2 \phi_y \bigg\},\\
        &\geq (\beta - d) \phi.
    \end{align*}
    It follows that $\phi(x,y)$ is a super-solution to \eqref{pde_beta_CP} by choosing $\beta \geq d$.

    Let $\delta > 0$ and define $V_\delta := V + \delta \phi$. Then $V_\delta$ is a viscosity super-solution to \eqref{pde_beta_CP}. Furthermore, given the growth condition of $U$ and $V$, we have
    \begin{align} \label{infty_behavior}
        &\lim_{\substack{|x|+|y|\to\infty}} \sup_{t \in [0,T]} (U_\mu-V_\delta) (x,y,t)\\
        &= \lim_{\substack{|x|+|y|\to\infty}} \sup_{t \in [0,T]} \bigg(U - \frac{\mu}{t} - V - \delta \big(e^{2D (x+y)} + e^{2D (x-y)} + e^{2D (-x+y)} + e^{2D (-x-y)} \big)\bigg)\\
        &\leq \lim_{\substack{|x|+|y|\to\infty}} \sup_{t \in [0,T]} \bigg(U-V - \delta \big(e^{2D (|x|+|y|)} \big)\bigg)
        = -\infty.
    \end{align}

    \textbf{Step 3:} Our goal is to find the global maximum of $U_\mu-V_\delta$ on $[0,T] \times \mathbb{R}^2$ and show that it is non-positive by contradiction. By sending $\delta, \mu \rightarrow 0$, we complete the proof.

    By \eqref{infty_behavior}, there exists a constant $K > 0$ such that
    $$
    (U_\mu-V_\delta)(x,y,t) < 0, \quad \forall x, y \in [-K,K]^c, \quad \forall t \in [0,T].
    $$
    Moreover, \eqref{CP_initial_condition} implies that $\lim_{t \searrow 0}\sup_{\substack{(x,y) \in \mathbb{R}^2}} (U_\mu-V_\delta)(x,y,t) = -\infty$.
    
    Hence, there exist a compact set $\mathcal{K} \in \mathbb{R}^2$, and $t' > 0$ such that
    \begin{equation} \label{max_compact}
    \sup_{\substack{t \in [0,T] \\ (x,y) \in \mathbb{R}^2}} (U_\mu-V_\delta)(x,y,t) = \max_{\substack{t \in [t',T] \\ (x,y) \in \mathcal{K}}} (U_\mu-V_\delta)(x,y,t).
    \end{equation}
    The right-hand side of \eqref{max_compact} is attained by the Extreme Value Theorem for upper semi-continuous functions.

    Suppose $M_0 := \max_{\substack{t \in [t',T] \\ (x,y) \in \mathcal{K}}} (U_\mu-V_\delta)(x,y,t) > 0$. Then 
    \begin{equation} \label{max_nobdy}
        0 < M_0 = \max_{\substack{t \in [t',T) \\ (x,y) \in \mathcal{K}}} (U_\mu-V_\delta)(x,y,t)
    \end{equation}
    by noticing $(U_\mu-V_\delta)(x,y,T) \leq (U-V)(x,y,T) \leq 0$.
    Additionally, $\forall \epsilon > 0$, we define
    \begin{align}
        \Phi_\epsilon(t,s,x_1,y_1,x_2,y_2) &:= U_\mu(t,x_1,y_1) - V_\delta(s,x_2,y_2) - \phi_\epsilon(t,s,x_1,y_1,x_2,y_2),\label{fun_Phi}\\
        \phi_\epsilon(t,s,x_1,y_1,x_2,y_2) &:= \frac{1}{2\epsilon}\Bigg[(t-s)^2+(x_1-x_2)^2+(y_1-y_2)^2\Bigg].\label{fun_phi}  
    \end{align}
    Since $\Phi_\epsilon(t,s,x_1,y_1,x_2,y_2)$ is upper semi-continuous, the Extreme Value Theorem ensures that it attains a local maximum, denoted as $M_\epsilon$ on the compact set $[t',T]^2 \times \mathcal{K}^2$ at $(t^\epsilon,s^\epsilon,x_1^\epsilon,y_1^\epsilon,x_2^\epsilon,y_2^\epsilon)$. That is, $M_\epsilon = \Phi_\epsilon(t^\epsilon,s^\epsilon,x_1^\epsilon,y_1^\epsilon,x_2^\epsilon,y_2^\epsilon)$. Furthermore, we have the following lemma.
    \begin{lemma}
        $M_\epsilon \rightarrow M_0$ and $\phi_\epsilon(t^\epsilon,s^\epsilon,x_1^\epsilon,y_1^\epsilon,x_2^\epsilon,y_2^\epsilon) \rightarrow 0$ as $\epsilon \rightarrow 0$.
    \end{lemma}
    \begin{proof}
        Notice that
        \begin{align}
            0 < M_0 &\leq M_\epsilon = U_\mu(t^\epsilon,x_1^\epsilon,y_1^\epsilon) - V_\delta(s^\epsilon,x_2^\epsilon,y_2^\epsilon) - \phi_\epsilon(t^\epsilon,s^\epsilon,x_1^\epsilon,y_1^\epsilon,x_2^\epsilon,y_2^\epsilon) \label{ineq_phi}\\
            &\leq U_\mu(t^\epsilon,x_1^\epsilon,y_1^\epsilon) - V_\delta(s^\epsilon,x_2^\epsilon,y_2^\epsilon), \label{ineq_nophi}
        \end{align}
        where \eqref{ineq_phi} follows by fixing $s=t, x_1=x_2, y_1=y_2$. By the Bolzano–Weierstrass theorem, the bounded sequence $\big((t^\epsilon,s^\epsilon,x_1^\epsilon,y_1^\epsilon,x_2^\epsilon,y_2^\epsilon)\big)_\epsilon$ converges, up to a subsequence, to $(\Bar{t},\Bar{s},\Bar{x}_1,\Bar{y}_1,\Bar{x}_2,\Bar{y}_2) \in [t',T]^2 \times \mathcal{K}^2$ as $\epsilon \rightarrow 0$.
        Additionally, since $U_\mu(t^\epsilon,x_1^\epsilon,y_1^\epsilon) \allowbreak - V_\delta(s^\epsilon,x_2^\epsilon,y_2^\epsilon)$ is upper semi-continuous, the Extreme Value Theorem ensures that it is bounded above on the compact set $[t',T]^2 \times \mathcal{K}^2$. It then follows that $\Bar{t} = \Bar{s}, \Bar{x}_1 = \Bar{x}_2, \Bar{y}_1 = \Bar{y}_2$. Otherwise, \eqref{fun_phi} implies that $\lim_{\epsilon \rightarrow 0}\phi_\epsilon(t^\epsilon,s^\epsilon,x_1^\epsilon,y_1^\epsilon,x_2^\epsilon,y_2^\epsilon) = \infty$, contracting with \eqref{ineq_phi}. Sending $\epsilon \rightarrow 0$ in \eqref{ineq_nophi}, we obtain $M_0 \leq (U_\mu - V_\delta)(\Bar{t},\Bar{x}_1,\Bar{y}_1) \leq M_0$. Therefore, $M_0 = (U_\mu - V_\delta)(\Bar{t},\Bar{x}_1,\Bar{y}_1)$ with $(\Bar{t},\Bar{x}_1,\Bar{y}_1) \in [t',T) \times \mathcal{K}$ by \eqref{max_nobdy}. Sending $\epsilon \rightarrow 0$ again in \eqref{ineq_phi} and \eqref{ineq_nophi}, we have
        $$
        M_0 \leq \lim_{\epsilon \rightarrow 0} M_\epsilon = M_0 -\lim_{\epsilon \rightarrow 0} \phi_\epsilon(t^\epsilon,s^\epsilon,x_1^\epsilon,y_1^\epsilon,x_2^\epsilon,y_2^\epsilon) \leq M_0.
        $$
        Thus, $M_\epsilon \rightarrow M_0$ and $\phi_\epsilon(t^\epsilon,s^\epsilon,x_1^\epsilon,y_1^\epsilon,x_2^\epsilon,y_2^\epsilon) \rightarrow 0$ as $\epsilon \rightarrow 0$.
    \end{proof}

    \textbf{Step 4:}
    We start this step by re-writing the PDE \eqref{pde_beta_CP} in a more compact form and introducing Ishii's lemma.
    Denote $z = (x,y)$ and 
    $$
    H(p,R) = \sup_{\lambda} \bigg[b(\lambda) \cdot p + \frac{1}{2} \mathrm{tr} [QQ^T(\lambda) \cdot R]\bigg],
    $$
    where $p \in \mathbb{R}^2$, $R \in \mathcal{S}_2$(the set of symmetric $2 \times 2$ matrices), $b(\lambda) = \begin{pmatrix}
        \sigma \lambda, & \frac{c}{2}\lambda^2
    \end{pmatrix}$ and $QQ^T(\lambda) = \begin{pmatrix}
        \sigma^2 & \sigma c \lambda\\
        \sigma c \lambda & c^2 \lambda^2
    \end{pmatrix}$.
    Then, $U_\mu(z,t)$ (resp. $V_\delta(z,t)$) is a u.s.c. viscosity sub-solution(resp. l.s.c. viscosity super-solution) to the PDE \eqref{pde_beta_CP}, which can be rewritten as:
    \begin{equation}
        -w_t + \beta w - H(D_z w, D_z^2 w) = 0.
    \end{equation}
    \begin{lemma} \label{lemma_ish}
        (Ishii's lemma) Let $U$ (resp. $V$) be a u.s.c. (resp. l.s.c.) function on $(0,T) \times \mathbb{R}^2$. Let $\phi \in C^{1,1,2,2}\big((0,T)^2 \times \mathbb{R}^2 \times \mathbb{R}^2\big)$ and let $(\Bar{t},\Bar{s},\Bar{z}_1,\Bar{z}_2) \in (0,T) \times \mathbb{R}^2 \times \mathbb{R}^2$ be a local maximum of $U(t,z_1) - V(s,z_2) - \phi(t,s,z_1,z_2)$. Then, $\forall \eta > 0$, there exist $M, N \in \mathcal{S}_2$ such that
        \begin{align} \label{lim_jet}
            \bigg(\frac{\partial \phi}{\partial t}(\Bar{t},\Bar{s},\Bar{z}_1,\Bar{z}_2), D_{z_1} \phi(\Bar{t},\Bar{s},\Bar{z}_1,\Bar{z}_2), M\bigg) \in \Bar{\mathcal{P}}^{2,+} U(\Bar{t},\Bar{z}_1),\\
            \bigg(-\frac{\partial \phi}{\partial s}(\Bar{t},\Bar{s},\Bar{z}_1,\Bar{z}_2), -D_{z_2} \phi(\Bar{t},\Bar{s},\Bar{z}_1,\Bar{z}_2), N\bigg) \in \Bar{\mathcal{P}}^{2,-} V(\Bar{s},\Bar{z}_2),
        \end{align}
        and
        \begin{align} \label{ineq_ish}
            \begin{pmatrix}
                M & 0 \\
                0 & -N
            \end{pmatrix}
            \leq D^2_{\substack{z_1,z_2}} \phi(\Bar{t},\Bar{s},\Bar{z}_1,\Bar{z}_2) + \eta \bigg(D^2_{\substack{z_1,z_2}} \phi(\Bar{t},\Bar{s},\Bar{z}_1,\Bar{z}_2)\bigg)^2.
        \end{align}
    \end{lemma}
    \begin{remark} \label{remark_ish}
        We apply Lemma \ref{lemma_ish} with $z = (x,y)$ and
        $$
        \phi(t,s,z_1,z_2) = \frac{1}{2\epsilon}\Bigg[(t-s)^2+(z_1-z_2)^2\Bigg] = \frac{1}{2\epsilon}\Bigg[(t-s)^2+(x_1-x_2)^2+(y_1-y_2)^2\Bigg].
        $$
        Direct calculation yields
        $$
        \frac{\partial \phi}{\partial t} = -\frac{\partial \phi}{\partial s} = \frac{t-s}{\epsilon}, \quad D_{z_1} \phi = -D_{z_2} \phi = \begin{pmatrix}
            \frac{x_1-x_2}{\epsilon} & \frac{y_1-y_2}{\epsilon}
        \end{pmatrix},
        $$
        $$
        D^2_{\substack{z_1,z_2}} \phi = \frac{1}{\epsilon} \begin{pmatrix}
            I_2&-I_2\\
            -I_2&I_2\\
        \end{pmatrix}, \quad \bigg(D^2_{\substack{z_1,z_2}} \phi\bigg)^2 = \frac{2}{\epsilon^2} \begin{pmatrix}
            I_2&-I_2\\
            -I_2&I_2\\
        \end{pmatrix}.
        $$
        Furthermore, by setting $\eta = \epsilon$ in \eqref{ineq_ish}, we obtain
        \begin{equation} \label{ineq_MN}
            \begin{pmatrix}
                M&0\\
                0&-N
            \end{pmatrix} \leq \frac{3}{\epsilon}\begin{pmatrix}
                I_2&-I_2\\
                -I_2&I_2
            \end{pmatrix}.
        \end{equation}
        Let $\Sigma(\lambda) = \begin{pmatrix}
            QQ^T(\lambda) & QQ^T(\lambda)\\
            QQ^T(\lambda) & QQ^T(\lambda)
        \end{pmatrix}$. 
        It follows that
        \begin{align} \label{trac_ineq}
            \mathrm{tr} \bigg(QQ^T(\lambda)(M-N)\bigg) &= \mathrm{tr} \bigg(\Sigma(\lambda)\begin{pmatrix}
                M & 0 \\
                0 & -N
            \end{pmatrix}\bigg)\\
            &\leq \frac{3}{\epsilon}\mathrm{tr} \bigg(\Sigma(\lambda)\begin{pmatrix}
                I_2&-I_2\\
                -I_2&I_2
            \end{pmatrix}\bigg)\\
            &= \frac{3}{\epsilon}\mathrm{tr} \big((Q(\lambda) - Q(\lambda))(Q(\lambda) - Q(\lambda))^T\big) = 0.
        \end{align}
    \end{remark}
    By \eqref{lim_jet}, $\forall \epsilon > 0$, we have
    \begin{align}
        -\frac{t^\epsilon-s^\epsilon}{\epsilon} + \beta U_\mu(t^\epsilon,z_1^\epsilon) - H(
        \begin{pmatrix}
            \frac{x_1^\epsilon-x_2^\epsilon}{\epsilon} & \frac{y_1^\epsilon-y_2^\epsilon}{\epsilon}
        \end{pmatrix}
        , M) &\leq 0,\\
        -\frac{t^\epsilon-s^\epsilon}{\epsilon} + \beta V_\delta(s^\epsilon,z_2^\epsilon) - H(
        \begin{pmatrix}
            \frac{x_1^\epsilon-x_2^\epsilon}{\epsilon} & \frac{y_1^\epsilon-y_2^\epsilon}{\epsilon}
        \end{pmatrix}
        , N) &\geq 0.
    \end{align}
    Thus, by \eqref{ineq_nophi}, we have
    \begin{align}
        0 < \beta M_0 &\leq \beta \bigg[U_\mu(t^\epsilon,z_1^\epsilon) - V_\delta(s^\epsilon,z_2^\epsilon)\bigg]\\ 
        &\leq H(
        \begin{pmatrix}
            \frac{x_1^\epsilon-x_2^\epsilon}{\epsilon} & \frac{y_1^\epsilon-y_2^\epsilon}{\epsilon}
        \end{pmatrix}
        , M) - H(
        \begin{pmatrix}
            \frac{x_1^\epsilon-x_2^\epsilon}{\epsilon} & \frac{y_1^\epsilon-y_2^\epsilon}{\epsilon}
        \end{pmatrix}
        , N)\\
        &\leq \sup_\lambda \bigg[\frac{1}{2}\mathrm{tr}[QQ^T(\lambda)(M-N)]\bigg]\\
        &\leq 0,
    \end{align}
    where the last inequality follows from \eqref{trac_ineq}.
    Given $\beta >0$, we get a contradiction. Sending $\delta, \mu \rightarrow 0$ completes the proof.

\section{Auxiliary Lemmas}
\begin{lemma} \label{lemma_conts_y}
    {Let $I \subseteq \mathbb{R}$ be an open interval.} Suppose the functions $H(x,y,t) \in C(\mathbb{R} {\times I} \times [0,T])$ and $U(x,y) \in C(\mathbb{R} \times I)$, 
    satisfy 
    $$
    \left|H(x,y,t)\right|, \left|U(x,y)\right| \leq c_6(1+e^{c_6|x|}),
    $$
    {where $c_6$ is locally uniformly bounded with respect to $y \in I$}, then the classical solution $f(x,y,t)$ to the  following PDE
    \begin{equation} \label{PDE_w&H}
    \begin{cases}
    f_t+\frac{\sigma^2}{2} f_{x x}=H(x,y,t),\quad \text{for } (x, y, t) \in \mathbb{R} \times I \times[0, T),\\
    f(x,y,T) = U(x,y)
    \end{cases}
    \end{equation}
    is continuous on $\mathbb{R} \times I \times [0,T]$.
\end{lemma}
\begin{proof}[Proof of Lemma \ref{lemma_conts_y}]
The solution admits the following representation:
    \begin{align}
    f(x,y,t) &=
    \int_{\mathbb{R}} \mathcal{K}(T-t,\, x,\, \xi) \; U(\xi, y) \, d\xi - \int_{t}^{T} \int_{\mathbb{R}} \mathcal{K}(t'-t,\, x,\, \xi) \; H(\xi, y, t') \, d\xi \, dt',
    \end{align}
    for any $0 \leq t < T$, where
    $$
    \mathcal{K}(\Delta t, x, \xi) = \frac{1}{\sqrt{2\pi \sigma^2 \Delta t}}
    \, \exp\!\Bigl( -\frac{(x-\xi)^2}{2\sigma^2 \Delta t} \Bigr), \qquad \forall  \Delta t > 0.
    $$

    To establish the continuity of $f(x, y, t)$, it suffices to prove the following two results:
    
    (1). $f(x,y,t)$ is continuous in $y$. 
    
    This result follows immediately from the Dominated Convergence Theorem by noticing the locally uniformly exponential growth of $U(x,y)$ and $H(x,y,t)$ in $x$. 
    
    (2). $f(x,y,t)$ is locally uniformly continuous in $(x, t)$.   
    
    When $s<T$, $x'\in \mathbb{R}$, $y' \in I$, $f(x,y,t)$ is {uniformly continuous} in $(x,t)$ around $(x', y', s)$ due to the classical $W^{2, 1}_p$ estimate and the Sobolev Embedding Theorem. 
     When $s=T$, $x' \in \mathbb{R}$, $y' \in I$, we have 
    \begin{align}
    f(x,y',t) - f(x',y',T) &=
    \int_{\mathbb{R}} \mathcal{K}(T-t,\, x,\, \xi) \; \left(U(\xi, y')- U(x',y')\right) \, d\xi\\
    &\qquad - \int_{t}^{T} \int_{\mathbb{R}} \mathcal{K}(t'-t,\, x,\, \xi) \; H(\xi, y', t') \, d\xi \, dt'.
    \end{align}

    First, notice that the second term on the right-hand side converges to 0 locally uniformly with respect to $y'$ when $(x, t)$ converges to $(x', T)$. This is due to the exponential growth assumption of $H$.

    Next, we estimate the first term on the right-hand side. 
    {Since $U$ is continuous, by Heine–Cantor theorem, it is uniformly continuous when $|x'-x|, |y'-y|, |T-t|\leq 1$.}
    {Thus,} for any given $\epsilon>0$, {there exists a $0 < \delta \leq 1$, locally independent of $y'$, such that 
    $$
    |U(x,y') - U(x',y')| < \frac{\epsilon}{2}, \quad \forall \, |x-x'| < \delta.
    $$
    }
    Therefore, we estimate the first term on the right-hand side by splitting the integral into two regions: $|\xi -x'| {\geq} \delta$ and $|\xi -x'| {<} \delta$.
    
    First, we consider the case $|\xi -x'| {\geq} \delta$. For any $|x'-x|<\frac{\delta}{2}$ we have $|x-\xi| \geq -|x-x'| + |x'-\xi| \geq \frac{|x'-\xi|}{2}$. Setting 
    $z = \frac{\xi-x'}{\sqrt{T-t}}$, we obtain
    \begin{align}
        &\int_{|\xi -x'| {\geq} \delta} \mathcal{K}(T-t,\, x,\, \xi) \; \left|U(\xi, y')- U(x',y')\right| \, d\xi\\
        = &\int_{|\xi -x'| {\geq} \delta} \frac{1}{\sqrt{2\pi \sigma^2 (T-t)}}
        \, \exp\!\Bigl( -\frac{(x-\xi)^2}{2\sigma^2 (T-t)} \Bigr)\; \left|U(\xi, y')- U(x',y')\right| \, d\xi\\
        \leq& \int_{|\xi -x'| {\geq} \delta} \frac{1}{\sqrt{2\pi \sigma^2 (T-t)}}
        \, \exp\!\Bigl( -\frac{(x'-\xi)^2}{8\sigma^2 (T-t)} \Bigr)\; \left|U(\xi, y')- U(x',y')\right| \, d\xi\\
        = & \int_{|z| {\geq} \frac{\delta}{\sqrt{T-t}}} \frac{1}{\sqrt{2\pi \sigma^2}}\exp\!\Bigl( -\frac{z^2}{8\sigma^2} \Bigr)\; \left|U(x'+ \sqrt{T-t} z, y')- U(x',y')\right| \, dz.
    \end{align}
    By the exponential growth condition on $U$, this integral is no more than $\frac{\epsilon}{2}$ when $|x'-x|, |T-t|$ are sufficiently small. Moreover, this criterion is locally independent of $y'$.    

   Next, we consider the case $|\xi -x'| {<} \delta$. We have   
   \begin{align}
       &\int_{|\xi -x'| {<}  \delta} \mathcal{K}(T-t,\, x,\, \xi) \; \left|U(\xi, y')- U(x',y')\right| \, d\xi\\
       {<} & {\frac{\epsilon}{2}} \int_{\mathbb{R}} \mathcal{K}(T-t,\, x,\, \xi) \, d\xi\\
       = &\frac{\epsilon}{2}.
   \end{align}
    In summary, $\int_{\mathbb{R}} \mathcal{K}(T-t,\, x,\, \xi) \; \left|U(\xi, y')- U(x',y')\right| \, d\xi {<} \epsilon$ when $(x, t)$ is sufficiently close to $(x', T)$ and this criterion is locally independent of $y'$. 
    
    This completes the proof for the continuity of $f$.  
\end{proof}

\end{document}